\documentclass{article}
 \usepackage{setspace}

\usepackage{amssymb}
\usepackage{amsmath}
\usepackage{amsthm}
\usepackage{array}
\usepackage[hidelinks]{hyperref}
\usepackage{graphicx}
\usepackage{dsfont}
\usepackage[numbers,sort&compress]{natbib}
\usepackage{authblk}
\usepackage{color}
\numberwithin{equation}{section}

\newtheorem{Proposition}{Proposition}
\newtheorem{Theorem}[Proposition]{Theorem}
\newtheorem{Lemma}[Proposition]{Lemma}

\newtheorem{Remark}[Proposition]{Remark}
\newtheorem{Note}[Proposition]{Note}

\def\RR{\mathbb{R}}
\def\CC{\mathbb{C}}
\def\ZZ{\mathbb{Z}}
\def\NN{\mathbb{N}}

\title{Non-perturbative Solution of the 1d Schr\"odinger Equation Describing Photoemission from a Sommerfeld model Metal by an Oscillating Field}
\begin{document}

\date{}
\author[1]{Ovidu Costin\thanks{costin.9@osu.edu}}
\author[1]{Rodica Costin\thanks{costin.10@osu.edu}}
\author[2]{Ian Jauslin\thanks{ian.jauslin@rutgers.edu}}
\author[2,3]{Joel L. Lebowitz\thanks{lebowitz@math.rutgers.edu}}
\affil[1]{Department of Mathematics, The Ohio State University}
\affil[2]{Department of Mathematics, Rutgers University}
\affil[3]{Department of Physics, Rutgers University}
\maketitle

\begin{abstract}
  We analyze non-perturbatively the one-dimensional Schr\"odinger equation describing the emission of electrons from a model metal surface by a classical oscillating electric field. Placing the metal in the half-space $x\leqslant 0$, the Schr\"odinger equation of the system is $i\partial_t\psi=-\frac12\partial_x^2\psi+\Theta(x) (U-E x \cos\omega t)\psi$, $t>0$, $x\in\mathbb R$, where $\Theta(x)$ is the Heaviside function and  $U>0$ is the effective confining potential (we choose units so that $m=e=\hbar=1$). The amplitude $E$ of the external electric field  and the frequency $\omega$ are arbitrary. We prove existence and uniqueness of classical solutions of the Schr\"odinger equation for general initial conditions $\psi(x,0)=f(x)$, $x\in\mathbb R$. When the initial condition is in $L^2$ the evolution is unitary and the wave function goes to zero at any fixed $x$ as $t\to\infty$. To show this we prove a RAGE type theorem and show that the discrete spectrum of the quasienergy operator is empty. To obtain positive electron current we consider non-$L^2$ initial conditions containing an incoming beam from the left. The beam is partially reflected and partially transmitted for all $t>0$. For these initial conditions we show that the solution approaches in the large $t$ limit a periodic state that satisfies an infinite set of equations formally derived, under the assumption that the solution is periodic, by Faisal, et. al [Phys. Rev. A {\bf 72}, 023412 (2005)].  Due to a number of pathological features of the Hamiltonian (among which unboundedness in the physical as well as the spatial Fourier domain) the existing methods to prove such results do not apply, and we introduce new, more general ones. The actual solution exhibits a very complex behavior, as seen both analytically and numerically.  It shows a steep increase in the current as the frequency passes a threshold value $\omega=\omega_c$, with $\omega_c$ depending on the strength of the electric field. For small $E$, $\omega_c$ represents the threshold in the classical photoelectric effect, as described by Einstein's theory.
\end{abstract}
\vfill
\eject

\tableofcontents

\vfill
\eject

\section{Introduction}
\subsection{ Physical setting}

The emission of electrons from a metal surface induced by the application of an external electric field is a problem of continuing theoretical and practical interest \cite{HKK06,SKH10,BGe10,KSH11,KSe12,THH12,HSe12,PPe12,HBe12,PSe14,HWR14,EHe15,BBe15,YSe16,FSe16,RLe16,FPe16,LJ16,HKe17,SSe17,PHe17,Je17,WKe17,KLe18,LCe18,SMe18}.
It was first fully analyzed for constant electric field  using the ``new mechanics'' by Fowler and Nordheim (FN) in 1928 \cite{FN28}. They considered the Sommerfeld model of quasi-free electrons confined to a metal
 occupying the entire half-space $x<0$  by an effective step potential $U$. The metal is filled with electrons up to a Fermi level $\mathcal E_F$,  neglecting the small number of thermal electrons at room temperatures.
This gives the work function $W:=U-\mathcal E_F$ , i.e. $W$ is the minimum amount of energy necessary to take an electron out of the metal.

Applying a constant external electric field $E$ for $x>0$, see Figure \ref{fig:FN}, an electron in the Fermi sea moving in the positive $x$-direction, described by a plane wave $e^{ikx}$, $k>0$, can then tunnel out of the metal (we use units in which $\hbar=m=e=1$).

To describe this system FN considered the Schr\"odinger equation
\begin{equation}\label{eq1}
  i\partial_t\psi=-\frac12\partial_x^2\psi+\Theta(x)(U-Ex)\psi
\end{equation}
where $\Theta(x)$ is the Heaviside function, equal to $1$ if $x>0$ and $0$ otherwise.
To compute the stationary current observed after the field has been on for a while, FN  made the Ansatz that $\psi(x,t)$ is a generalized eigenfunction of \eqref{eq1}
\begin{equation}\label{1p2}
\psi(x,t)=e^{-\frac{ik^2}2t}\phi_E(x)
\end{equation}
with $\phi_E$ satisfying the equation
\begin{equation}
  \frac{k^2}2\phi_E=\frac12\partial_x^2\phi_E-\Theta(x)(U-Ex)\phi_E
  .
  \label{eq1E}
\end{equation}

The requirement that there be only one incoming wave from the left, given by $e^{i k x}, k>0$, for $x<0$ and only outgoing electrons for $x>0$, as well as that $\phi_E(x)$ and its derivative be continuous at $x=0$, and that $\phi_E(x)$ be bounded as $|x|\to\infty$, gave $\phi_E(x)=e^{i k x} +R_E e^{-i k x}$ for $x<0$ and an Airy function expression for $x>0$.

\begin{figure}
\hfil\includegraphics[width=8cm]{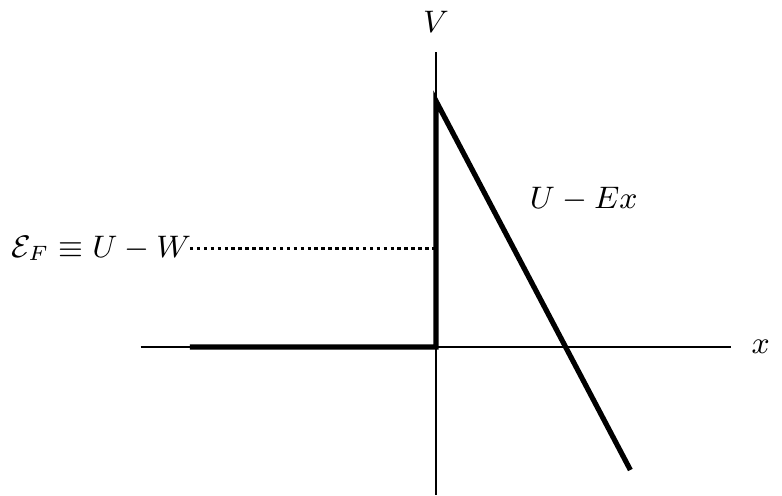}
\caption{The potential considered by Fowler and Nordheim. $x<0$ corresponds to the region inside the metal and $x>0$ corresponds to the vacuum outside.}
  \label{fig:FN}
\end{figure}

The FN computation is still the basic ingredient for the analysis of constant field currents experiments at present \cite{FN28,FKS05,HKK06,KSH11,YGR11,Ba06,KSe12b,PA12,YHe13,CPe14,ZL16,Fo16,Je17,KLe18,LZZ21}.
Their analysis does not consider the initial state of the system when the field is turned on. To check the validity of the FN  ansatz \eqref{1p2} we
recently revisited the FN setup by solving \eqref{eq1} for general initial values of $\psi(x,0)$.\,We showed that for all $\psi(x,0)$ representing an incoming beam $e^{ikx}$ \cite{CCe20} plus some square integrable function, $\psi(x,t)$ converges to the FN solution when $t\to\infty$. The asymptotic approach behaves like $t^{-\frac32}$. We considered in particular the initial state corresponding to a solution of  \eqref{eq1E} when $E=0$:
\begin{equation}
  \psi(x,0)=\phi_0(x)=
  \left\{\begin{array}{ll}
    e^{ikx}+R_0e^{-ikx}
    &\mathrm{for\ }x\leqslant 0
    \\
    T_0e^{-\sqrt{2U-k^2}x}
    &\mathrm{for\ }x> 0
  \end{array}\right.
  ,\quad
  R_0=\frac{ik+\sqrt{2U-k^2}}{ik-\sqrt{2U-k^2}}
  ,\quad
  T_0=\frac{2ik}{ik-\sqrt{2U-k^2}}
  .
  \label{init}
\end{equation}

\bigskip

 {\bf Time-periodic electric field and the photoelectric effect.} In the present work, we consider a setup similar to that of FN, except that the external field $E$ is taken to be periodic in time with period $\frac{2\pi}\omega$.
More precisely, we consider solutions of the equation
\begin{equation}
i\partial_t\psi =-\frac12\partial_x^2\psi+\Theta(x)(U-E x \cos\omega t)\psi,\ \ \ t>0
  \label{schrodinger2}
\end{equation}
with an initial value $\psi(x,0)$.
Physically, this can represent, depending on $\omega$, a great variety of situations ranging from  an alternating field produced by a mechanical generator to one produced by shining a laser on the metal surface.

For small values of $\omega$ the situation is in some ways similar to the constant field case with electrons tunneling through the (oscillating) barrier, although the limit $\omega\to 0$ in \eqref{schrodinger2} is very singular. 
For larger $\omega$, the situation is expected to be similar to that of the photoelectric effect, where light shining on a metal surface causes the almost instantaneous emission of electrons with a well-defined maximum kinetic energy $K$, given by the Einstein formula $K=\omega-W$ (recall that $\hbar=1$ in our units). Here of course we do not consider discrete photons, since   \eqref{schrodinger2} represents the electric field classically. It is expected however that the discrete jumps will show up as resonances, see \cite{CCL18}. Something like this is indeed the case for weak fields \cite{Wo35}. For larger fields one has to add to $W$ the ponderomotive energy $\frac{E^2}{4\omega^2}$ \cite{Wo35} of the electron in the oscillating field, see Figure \ref{Fig3} in the Appendix. There is a vast physical literature on this topic: For a comprehensive review see \cite{DPe20} and references therein.

\subsection{Mathematical setting.} 
From a mathematical point of view, the existence of solutions of \eqref{schrodinger2} with appropriate physical initial conditions which remain bounded and behave  in a physical way for all $x$ and $t$ is not obvious. In the physics literature, Faisal et al. \cite{FKS05}  considered periodic solutions of \eqref{schrodinger2} for general periodic fields $E(t)=E(t+2\pi/\omega)$ and, in analogy to the work of FN  sought solutions of \eqref{schrodinger2} in the form \footnote{Using the magnetic rather than the length gauge.}
\begin{equation}
  \label{eq:Faisal-a}
  \psi(x,t)=e^{-\frac12 i k^2 t}\phi(x,t)
\end{equation}
where $\phi(x,t)$ is periodic in time and has a single incoming wave $e^{i k x}$ for $x<0$, $k>0$. The continuity conditions at $x=0$ then lead to an infinite set of linear equations for the time-Fourier coefficients of $\phi$. The existence of solutions for this infinite system was not proven. What Faisal \& al. did was to truncate the infinite set of equations and solve the truncated system numerically.

In this paper we  rigorously analyze the full time evolution of \eqref{schrodinger2} both for $L^2$ initial conditions as well as for  an incoming beam $e^{ikx}$ as in \eqref{init} plus other terms which do not contribute to the long time behavior. We then find that for $L^2$ initial conditions $\psi(x,t)$ decays pointwise at least at a rate $O(t^{-1/2})$. For this, we first obtain a RAGE-type  theorem for this time-dependent potential. In the case the initial condition contains an incoming wave as in \eqref{init} (plus possible $L^2$ perturbations), the solution converges at least at a rate $O(t^{-1/2})$ to the ansatz in \cite{FKS05}. It follows from our result that the infinite system of equations obtained by  Faisal \& al. has a solution. We limit our analysis to time-periodic fields of the form in \eqref{schrodinger2} but expect our results to extend to general periodic fields.

To obtain these results we derive an integral equation \eqref{eqcontr} for $\psi(x,0):=\psi_0(x)$, which we show to have a unique solution (Lemma\,\ref{lemma:contraction}). We also obtain a set of formulas \eqref{psiminus}, \eqref{psiplus}, \eqref{fomulapsix0}, that recover the full wave function $\psi(x,t)$ from $\psi_0$. The properties of $\psi_0$, and therefore of $\psi$, are derived from the integral equation that it solves.   By far the most delicate analysis concerns the long time behavior of the solution of the Schr\"odinger equation.

 Behind the apparent simplicity of the potential in \eqref{schrodinger2} lie a number of significant mathematical difficulties making the analysis particularly challenging. Among them: lack of smoothness, and the fact that the Hamiltonian is unbounded in a time dependent way both in physical domain and in momentum space (owing to  the unboundedness of the potential energy term and lack of continuity). As a result, the classical PDE toolkit does not apply. To overcome these difficulties, we develop new methods, described in \S\ref{S:Math}, which we  combine with the spectral measure theory of the underlying unbounded operators. Preliminary results, without proofs, were given in \cite{CCe20}; that paper also contains interesting, and rigorously controlled numerical findings about the solutions, see Appendix \ref{appendix:figs}.

\section{Main Results}

Denote
\begin{equation}\label{defD}
\mathcal{D}=H^2(\RR\setminus\{0\})\cap H^1(\RR)\cap \{f\, |\, xf\in L^2(\RR)\}
\end{equation}

\begin{Theorem}\label{Thzero}
    \begin{enumerate}
\item[(a)] The Hamiltonians $\mathcal{H}_t:=-\frac12\partial_x^2\psi+\Theta(x)(U-E x \cos\omega t)\psi$, densely defined on $C_0^\infty(\RR)$ have self-adjoint extensions on $\mathcal{D}$ for each fixed $t$.
    
  \item[(b)] 
  Assuming $L^2$ initial conditions, the Schr\"odinger evolution  in the model \eqref{eq1} is unitary.
  \end{enumerate}
\end{Theorem}

\begin{Theorem}\label{theo:existence}

 If the initial state $\psi(\cdot,0):=f$ is in $\mathcal{D} $, then \eqref{schrodinger2} has a unique solution $\psi(\cdot,t)\in\mathcal{D}$, and $\psi(x,t)$ is continuously differentiable in $t>0$.

\end{Theorem}

\begin{Theorem}[Long Time Behavior]\label{Thlongtime}
  (i) For initial conditions in a dense subset of $\mathcal{D}$ we have: for any compact set $A\subset \RR$ the long time behavior of solutions is\footnote{We believe that the actual behavior below is $O(t^{-3})$, but this results from difficult to calculate cancellations occurring in algebraically cumbersome expressions.}
  \begin{equation}\label{intpsisq}
  \int_A|\psi(x,t)|^2\, dx=O(t^{-1})\ \ \ \text{as\ }t\to\infty
  \end{equation}
  
\noindent   (ii) If the initial condition $\psi(\cdot,0)$ is in $\mathcal{D}$, then
   \begin{equation}
  \label{largetphi}
    \lim_{t\to\infty}\psi(x,t)=0  \end{equation}
  uniformly in $x$ in compact sets in $\RR$.
\end{Theorem}

\begin{Theorem}[Wave Initial Condition]\label{theo:2}
 For the initial state \eqref{init}  equation \eqref{schrodinger2} has a unique solution that is bounded, and
   $$\psi(x,t)\sim e^{-ik^2 t/2}\phi(x,t)\ \ \ \text{as\ }t\to\infty$$
 where $\phi$ is time-periodic of period $2\pi/\omega$. 
\end{Theorem}

\begin{Remark}
{\rm   In the proof of Theorem\,\ref{theo:2} we make an additional simplifying assumption: $U+\frac{E^2}{4\omega^2}$ is not an integer multiple of $\omega$, and neither is $U+\frac{E^2}{4\omega^2}-\frac{k^2}2$.
  We do this because these two special cases have a slightly different singularity structure, which would require small changes in the proof, which we will not belabor.
   The two exceptional cases above correspond to a marginal situation in which absorbing an integer number of photons raises the energy of the electron to exactly the ionization value.}
\end{Remark}

\begin{Remark}\label{Faisal}
{\rm In \cite{FKS05}, Faisal, Kami\'nski and Saczuk computed the periodic solutions of the Schr\"odinger equation \eqref{schrodinger2} with an incoming plane wave.
By Theorem \ref{theo:2}, the solution they computed must be the asymptotic solution $\phi$.}
\end{Remark}

 The rest of the article deals with proving these results.

\subsection{ Outline of the Mathematical Approach}\label{S:Math}
The external potential $\Theta(x)(U-E x\cos(\omega t))$ is unbounded both in the physical domain and, due to low regularity,  in spatial Fourier space. These issues are at the root of some of the more serious difficulties of this model. Since non-smoothness is localized at $x=0$, it is convenient to work with one-sided Fourier transforms, by means of which we obtain a left-to-right continuity  integral equation. Existence, uniqueness, regularity and unitarity are derived,  by more or less standard operator theory techniques, in \S \ref{S:proof 1b} from the Fourier transform of this equation.
   
   Specific information about the behavior of the system is obtained from the equation satisfied by $\psi(0,t)$, an equation of the form \eqref{eqcontr} below. The integral operator in this equation is quite involved.  The  high complexity of the equations governing the evolution of many quantities of interest represents another source of technical difficulties.

   By far the most delicate task in this model is finding the large time behavior of the system. The usual Laplace transform methods (see \cite{CLT10,CCL18} and references therein) cannot be used here because of their daunting algebraic complexity. Instead,  we introduce a number of new methods.

   In a nutshell, we rely on ``sampling'' the wave function at $t=t_n=n(1+r) \frac{2\pi}\omega, n\in\NN, r\in [0,1)$ which we use as coefficients of a generating function, which is analytic in $r$ in the open unit disk. This analyticity only requires exponential bounds on the growth of the wave function with respect to time, a type of bounds  which are not difficult to get from the integral equation it satisfies. This generating function satisfies a sequence of equation based on compact operators in a family of Banach spaces (a type of decomposition of the governing equation that also seems new).

   The type of singularities of the generating function on the unit {\it circle} determine, by means of asymptotics of Fourier coefficients, the long time behavior of the system (see \S \ref{PfLongt}). If these singularities are weaker than poles, then $L^2$ initial conditions result in decay of the wave function for large time, pointwise in $x$. The presence of poles has an equivalent reformulation as the existence of nontrivial discrete spectrum of a compact operator in (a sequence of) Banach spaces. We show that the discrete spectrum of the aforementioned compact operators is empty, a property which is equivalent to the absence of poles of the generating function, hence  of bound states of the associated quasi-energy operator. The analysis of bound states of the quasi-energy operator, always a nontrivial task, is especially delicate here, and to tackle it we resorted to a new approach relying on the theory of resurgence and transseries, cf. \S\ref{RAGE}, as well as techniques of determining the global analytic structure of functions from their Maclaurin series \cite{CX15}, see \S\ref{calcPh}--\S\ref{merosol}.

   The proof of Theorem \ref{Thzero} (a) is given in \S\ref{S:proof 1a};  The proof of Theorem \ref{Thzero} (b) is given in   \S\ref{S:proof 1b}. The proof of Theorem \ref{theo:existence} is given in \S\ref{S:ProofT2} and Theorems \ref{Thlongtime} and \ref{theo:2} are proved in \S\ref{PfLongt}.

   \section{The spatial Fourier transform of  \eqref{schrodinger2}}\label{deriving}

Before turning to the proofs of the main results, we reduce the Schr\"odinger equation \eqref{schrodinger2} to a system of integral equations, which are derived by taking one-sided (half-line) Fourier transforms of $\psi$, denoted by $\hat\psi_-$ and $\hat\psi_+$ (this is equivalent to taking a pair of Laplace transforms; see also the paper by Fokas \cite{Fo97}).

Denote $\psi_0(t)=\psi(0,t)$ and $\psi_{x.0}(t)=\partial_x\psi(0,t)$. Recall the notation $\psi(x,0):=f(x)$.

We show that these transforms are in $L^2$ when the initial condition is in $L^2$. For the initial condition \eqref{init}, the calculation is understood in the sense of distributions. After establishing the main equations we need, the proofs will rely on essentially reversing, rigorously, these steps.

We calculate $\psi_-$ for $x<0$ by taking the half-line Fourier transform of \eqref{schrodinger2} on $\RR_-$, and the solutions  $\psi_+$ for $x>0$ by taking the half-line Fourier transform on $\RR_+$. We then impose the matching condition $\psi_-(0-,t)=\psi_+(0+,t):=\psi_0(t)$ and $\partial_x\psi_-(0-,t)=\partial_x\psi_+(0+,t):=\psi_{x,0}(t)$. Then $\psi(x,t)=\Theta(-x)\psi_-(x,t)+\Theta(x)\psi_+(x,t)$ is a solution of \eqref{schrodinger2}. We write
$$\hat{\psi}(\xi,t)=\frac 1{\sqrt{2\pi}}\int_{-\infty}^\infty {\rm e}^{-ix\xi}\psi(x,t)\, dx:=\hat{\psi}_-(\xi,t)+\hat{\psi}_+(\xi,t)$$ where $\hat{\psi}_\pm$ are the half-line Fourier transform of $\psi_\pm$. 

\begin{Note}\label{Note}  {\it As usual, the Fourier $\mathcal{F}$ transform of an $L^2$ function $f$ on a noncompact region $\mathcal{R}$  is understood as an $L^2$ limit of Fourier integrals on increasing compact subdomains $\mathcal{R}_N$ such that $\displaystyle \bigcup_N\mathcal{R}_N=\mathcal{R}$. We have 
 $$\hat{f}:=\mathcal{F} f =\frac1{\sqrt{2\pi}}\,{\rm l.i.m}\int_{\mathcal{R}_N} e^{-i \xi x}f(x) dx$$
where we adopted the notation of \cite[p.11]{RS75b}: in $n$ dimensions $l.i.m.$ stands for the norm limit of the integral over a ball of radius $R$ as $R\to\infty$.

To avoid complicating the notation, when we are not performing operations with such integrals, we will simply write
$$\hat{f}:=\mathcal{F} f =\frac1{\sqrt{2\pi}}\,\int_{\mathcal{R}} e^{-i \xi x}f(x) dx$$
}\end{Note}

By \eqref{schrodinger2},  $\hat{\psi}_-(\xi,t)$ satisfies
  \begin{equation}\label{eqhatpsim}
  i\frac{\partial \hat{\psi}_-}{\partial t}=\frac12\xi^2 \hat{\psi}_--  \frac1{2\sqrt{2\pi}}  \psi_{x,0}(t)-i\xi   \frac1{2\sqrt{2\pi}} \psi_{0}(t)
  \end{equation}
where $\psi_{0}(t)=\psi(0_-,t)$ and $\psi_{x,0}(t)=\partial_x \psi(0_-,t)$, whose solution with initial condition $f$ is
\begin{equation}\label{hatpsimin}
 \hat{\psi}_-(\xi,t)= {\rm e}^{-i\xi^2 t/2}\left\{ C_-(\xi) + \frac1{2\sqrt{2\pi}}  \int_0^{t} {\rm e}^{i\xi^2 s/2}\ \left[  i\psi_{x,0}(s)-\xi \psi_{0}(s) \right]\, ds\right\}
 \end{equation}
where
\begin{equation}\label{Cminus}
C_-(\xi) =\frac 1{\sqrt{2\pi}}\int_{-\infty}^0 {\rm e}^{-iy\xi}f(y)\, dy
\end{equation}

Similarly, $\hat{\psi}_+$ satisfies
\begin{equation}\label{eqhatpsip}
i\frac{\partial \hat{\psi}_+}{\partial t}=-i E \cos \omega t\,  \frac{\partial \hat{\psi}_+}{\partial \xi} +\left( \frac 12\xi^2+U\right) \hat{\psi}_+ +  \frac1{2\sqrt{2\pi}}  \psi_{x,0}(t)+i\xi   \frac1{2\sqrt{2\pi}}  \psi_{0}(t)
\end{equation}
 where $\psi_{0}(t)=\psi(0_+,t)$ and $\psi_{x,0}(t)=\partial_x \psi(0_+,t)$ (since we will impose the matching conditions we denote the lateral limits at $0$ the same, to avoid an overburden of the notation), with the solution
  \begin{equation}\label{hatpsiplu}
    \hat{\psi}_+(\xi,t)
    =e^{-i\Phi(u,t)}  \left\{C_+(u)+\int_0^t \, e^{i\Phi(u,s)} \frac1{2\sqrt{2\pi}}  [ -i \psi_{x,0}(s)+\xi\psi_0(s)]\right\}\, ds
  \end{equation}
  where
  \begin{equation}\label{uPhi}
   u= \xi-\frac E\omega \sin\omega t,\ \ \  \ \ \Phi(u,t)= \frac12 u^2 t +\left( U+\frac{E^2}{4\omega^2}\right)t - \frac{E}{\omega^2} u \cos(\omega t) -\frac{E^2}{8\omega^3}\sin(2\omega t)
  \end{equation}
  and  
   \begin{equation}\label{Cplus}
   e^{i\frac{E}{\omega^2} \xi}C_+(\xi) =\frac 1{\sqrt{2\pi}}\int_0^\infty {\rm e}^{-iy\xi}f(y)\, dy
   \end{equation}
 Taking the inverse Fourier transform, we obtain that for $x<0$ the wave function $\psi:=\psi_-$ satisfies
   \begin{equation}\label{psiminus}
  \psi_-(x,t)=h_-(x,t)+ \frac{\sqrt i}2\sqrt{2\pi} \int_0^t ds\, \left(\psi_{x,0}(s)+\frac{ix}{t-s}\psi_0(s)\right)\frac{e^{\frac{ix^2}{2(t-s)}}}{\sqrt{t-s}}
  \end{equation}
 (note that the last term is a convergent improper integral) where
  \begin{equation}
    h_-(x,t)=\frac 1{\sqrt{2\pi i t}}\int_{-\infty}^0 dy \, f(y) \, {\rm e}^{\frac{i(x-y)^2}{2t}} 
        \label{hminus}
    .
  \end{equation}
  For $x>0$, $\psi:=\psi_+$ satisfies
   \begin{multline}\label{psiplus}
    \psi_+(x,t)
    =h_+(x,t)+\frac 1{2\sqrt{2\pi i}}\  \int_0^t \,ds\,  \frac{ -i \psi_{x,0}(s)+\frac E\omega \sin\omega t\, \psi_0(s)}{\sqrt{t-s}}\   \, e^{iF(x,s,t)}  \\
    +\frac1{2\sqrt{2\pi i}}\ \int_0^t \,ds\,\psi_0(s)\, \frac{\frac{E}{\omega^2} [\cos(\omega t) -\cos(\omega s)]+x}{(t-s)^{3/2}}  \, e^{iF(x,s,t)}
 \end{multline}
 where
  \begin{equation}
    h_+(x,t)=    \frac1{\sqrt{2\pi i t }} {e^{ix\frac{E}{\omega}\sin(\omega t)-i(U+\frac{E^2}{4\omega^2})t+i\frac{E^2}{8\omega^3}\sin(2\omega t)  }}\int_0^{\infty }dy\, f(y)\, e^{\frac{i\left[x-y-\frac{E}{\omega^2}(1-\cos(\omega t))\right]^2}{2t}}
    .
    \label{hplus}
  \end{equation}
and
 \begin{equation}\label{defFxst}
 F(x,s,t)=x \frac E\omega \sin\omega t -i\left(U+\frac{E^2}{4\omega^2}\right)(t-s) -i\frac{E^2}{8\omega^3}\left[\sin(2\omega t)-\sin(2\omega s)\right]+\frac{\left[x+\frac{E}{\omega^2} (\cos\omega t -\cos\omega s)\right]^2}{2(t-s)}
 \end{equation}
    From \eqref{psiminus} and \eqref{hminus} we have 
  \begin{equation}\label{psim0}
  \psi_-(0,t)=h_-(0,t)+ \frac {\sqrt{i}}{2\sqrt{2\pi}}\int_0^t\, ds\, \psi_{x,0}(s)\frac 1{\sqrt{t-s}}
  +\frac12\psi_0(t) 
  \end{equation}
   where
  \begin{equation} \label{hminus0}
    h_-(0,t)=\frac 1{\sqrt{2\pi i t}}\int_{-\infty}^0 dy \, f(y) \, {\rm e}^{\frac{iy^2}{2t}}
\end{equation}
From  \eqref{psiplus}-\eqref{hplus} we have
   \begin{multline}\label{psiplus0}
    \psi_+(0,t)
    =h_+(0,t)+\frac1{2\sqrt{2\pi i}}\ \int_0^t \,ds  [ -i \psi_{x,0}(s)+\frac E\omega \sin\omega t\, \psi_0(s)]\ 
  \frac1{\sqrt{t-s}}\, e^{iF_0(s,t)} \\
  +\frac1{2\sqrt{2\pi i}}\, 
     \int_0^t\,ds\,  \psi_0(s)  \frac { E}{\omega^2}\frac{\cos \left( \omega
\,t \right) -\cos \left( \omega\,s \right)   }{(t-s)^{3/2}}\,  e^{iF_0(s,t)}  +\tfrac12\psi_0(t)
    \end{multline}
  where
   \begin{equation} \label{hplus0}
    h_+(0,t)=    \frac1{\sqrt{2\pi i t }} {e^{-i(U+\frac{E^2}{4\omega^2})t+i\frac{E^2}{8\omega^3}\sin(2\omega t)}}\int_0^{\infty }dy\, f(y)\, e^{\frac{i(y+\frac{E}{\omega^2}(1-\cos(\omega t)))^2}{2t}} 
  \end{equation}
and
 \begin{equation} \label{Fzero}
 F_0(s,t)=F(0,s,t)=-\left(U+\frac{E^2}{4\omega^2}\right)\, (t-s)+\frac{E^2}{8\omega^3}[\sin(2\omega t)-\sin(2\omega s)]+\frac{E^2 (\cos\omega t -\cos\omega s)^2}{2{\omega^4}(t-s)}
  \end{equation}
     Imposing the condition that $\psi_-(0,t)=\psi_0(t)$ in \eqref{psim0} and that $\psi_+(0,t)=\psi_0(t)$ in \eqref{psiplus0} we obtain a system of equations for $\psi_0$ and $\psi_{x,0}$
\begin{equation}
\psi_0(t)=2h_+(0,t)+ \mathcal{T}+ \frac1{{\sqrt{2\pi i}}}  \int_0^{t} \, ds\, \psi_0(s)\,g(s,t)\,{\rm e}^{iF_0(s,t)}
\end{equation}
where
$$\mathcal{T}= \frac1{{\sqrt{2\pi i}}}  \int_0^{t} \, \left[ -i \psi_{x,0}(s)\right]\,  \frac 1{\sqrt{t-s}}\,{\rm e}^{iF_0(s,t)} \,ds$$
and 
$$g(s,t)=\frac E\omega\frac{\sin \omega t}{{\sqrt{t-s}}}+\frac E{\omega^2} \frac{ \cos \omega t- \cos \omega s}{(t-s)^{3/2}}.$$

 The continuity of $\psi$ and its derivative imply
 
\begin{equation}
  \label{eq:relpsipsix}
  \psi_0(t)=2 h_{-}(0,t)+\sqrt{i/2\pi}\int_0^{t}(t-s)^{-1/2}\psi_{x,0}(s)ds
\end{equation}
which will be used in Lemma\,\ref{lemma:matching} below to eliminate $\psi_{x,0}$ from the equation ensuring the continuity of $\psi$ at $0$: $\psi_+(0,t)=\psi_-(0,t)=\psi_0(t)$,  which in Lemma\,\ref{lemma:contraction} is then shown to have a unique solution.

\bigskip

\section{Proof of Theorem\,\ref{Thzero}}\label{S:proof 1b}

\subsection{ A few more general results}
The unitary transformation $\psi\mapsto \varphi=U_t\psi$ given by
\begin{equation}
  \varphi_t(x)=e^{-ixA_t\Theta(x)}\psi_t(x)
  ,\quad \text{ where }
  A_t:=\int_0^t d\tau\ E_\tau=\frac{e_1}\omega\sin(\omega t)
\end{equation}
maps \eqref{schrodinger2} to the magnetic gauge representation,
\begin{equation}
  i\partial_t\varphi_t(x)=(i\partial_x-\Theta(x)A_t)^2\varphi_t(x)+\Theta(x)V\varphi_t(x)=:\mathcal{H}_{A;t}\phi_t
  .
\end{equation}
The quasi-energy operator $\mathcal{K}$  is defined on the domain
\begin{multline}
\mathcal{D}(\mathcal K) =\left\{\psi\in L^2(\mathbb{T}\times \RR)\cap AC(\mathbb{T}\times \RR):\partial_t\psi\in L^2(\mathbb{T}\times \RR), \right.\\ \left.
\partial_x\psi\in L^2(\mathbb{T}\times \RR),\partial_x^2\psi\in L^2(\mathbb{T}\times \RR), \partial_x\psi(\cdot,t) \in AC( \RR)\right\}
\end{multline}
where $\mathbb{T}$ is the torus $\RR/2\pi\ZZ$,  by
\begin{equation}
 \mathcal{K}= -i\partial_t+(i\partial_x-\Theta(x)A_t)^2+\Theta(x)V
\end{equation}
Let
\begin{equation}
\mathcal{D}_{A} =\left\{\psi\in L^2(\RR),\partial_x\psi\in L^2(\RR), \partial_x\psi(\cdot,t) \in AC( \RR),\partial_x^2\psi\in L^2(\RR)\right\}
\end{equation}
\begin{Proposition}
(i) For each $t$, $\mathcal{H}_{A;t}$ is self-adjoint on $\mathcal{D}_A$.

    (ii)  $\mathcal{K}$ is self-adjoint on $\mathcal{D}$.
\end{Proposition}
\begin{proof}   We only prove (ii); (i) is similar and simpler. 
  We rely on Rellich's theorem [see Kato], which we restate for convenience.
  \begin{Theorem}[Rellich] Let $T$ be selfadjoint. If $A$ is symmetric and $T$-bounded with $T-$bound smaller than 1, then $T+A$ is also selfadjoint.
  \end{Theorem}
  Here $T-$ bounded means that $D(T)\subset D(A)$ and for any $u\in D(T)$ we have
  $$\|Au\|\le a\|u\|+b\|Tu\|$$
  and $b$ is the $T-$bound. We take $T=-i\partial_t-\partial_x^2$, with $D(T)$ given in the proposition and $A=\mathcal K-T$. Clearly $A$ is symmetric.  We first note that $-i\partial_x \Theta=-i\Theta \partial_x -i\delta$ where $\delta$ is the Dirac distribution at zero. It is enough to show that  $\Theta \partial_x$ and $\delta$ are $T-$bounded with $b<1$. Indeed, the time-dependent coefficients are bounded and commute with the spatial part,  and  $\Theta V$ is $T-$bounded with $b=0$.
The rest is fairly standard.  We start with $\Theta\partial_x$ and note that $\|\Theta\partial_x u\|\le \|\partial_x u\|$, and, for $u\in D(\partial_x^2)$ (the domain in the proposition with $k=0$)
  \begin{multline}
    \label{eq:boundnabla}
    \|\partial_x u\|^2+\|u\|^2=\int_{\RR}(\xi^2+1)|\hat u|^2 d\xi=\int_{-n}^n (\xi^2+1)|\hat u|^2 d\xi+\int_{|u|>n}(\xi^2+1)|\hat u|^2 d\xi\\ \le 2n (n^2+1)\|\hat u\|^2+\frac{1}{n^2+1}\int_{\RR} (\xi^4+1)|\hat u|^2 d\xi= 2n (n^2+1)\|u\|^2+\frac{\|\partial_x^2 u\|^2}{n^2+1}
  \end{multline}
  and the rest is straightforward. We check now that $\delta$ is $\partial_x-$ bounded with bound one. Indeed,
  $$\left|\delta u\right|=\left|\frac{1}{\sqrt{2\pi}}\int_{\RR} \hat u(\xi)d\xi\right|$$
  
\end{proof}
\subsection{Proof of Part (a)} \label{S:proof 1a}
The Hamiltonians $\mathcal{H}_t$ and $\mathcal{H}_{A;t}$ are related by a unitary transformation; it remains to verify the transformation of domains which which is straightforward.

 \subsection{Proof of Part (b)} \label{S:proof 1b}
We prove this result in Fourier space. Consider $f\in\mathcal{D}_0\subset \mathcal{D}$, a dense set of initial conditions, such that $f$ is  $C^\infty$, exponentially decaying at infinity, and $f(0)=f'(0)=f''(0)=f'''(0)=0$.

We see in \eqref{hatpsiplu} that the half-line Fourier transform of $\psi$ for $x>0$
\begin{equation}
  \hat\psi_+=T_1+T_2
\end{equation}
with
\begin{equation}
  T_1=e^{-i\Phi(u,t)}C_+(u)
  ,\quad
  T_2=e^{-i\Phi(u,t)} \int_0^t \, e^{i\Phi(u,s)} g(s,\xi)\, ds
  ,\quad
  g(s,\xi)=  -i \psi_{x,0}(s)+\xi\psi_0(s)
\end{equation}
Let $c$ be a constant large enough so that $\partial_t(\Phi(u,t)+ct)>0$ for all $u$ (such a $c$ satisfies $c>\frac{E^2}{2\omega^2}-U$).
Integrating by parts twice in $T_2$ and using the fact that $\partial_sg(0,\xi)=  \partial_{ss}g(0,\xi)=0$ we obtain
\begin{multline}
T_2=e^{-i\Phi(u,t)} \int_0^t \,  \frac1i\frac{e^{-ics}g(s,\xi)}{\Phi_s(u,s)+c} \, \partial_s e^{i\Phi(u,s)+ics}ds \\
=-i\frac{g(t,\xi)}{\Phi_t(u,t)+c}+ie^{-i\Phi(u,t)} \int_0^t \,  \partial_s\left( \frac{e^{-ics}g(s,\xi)}{\Phi_s(u,s)+c}\right) \, e^{i\Phi(u,s)+ics}ds\\
=-i\frac{g(t,\xi)}{\Phi_t(u,t)+c} e^{-i\Phi(u,t)}  +i\int_0^t \,  \frac1{\Phi_s(u,s)+c}\partial_s\left( \frac{e^{-ics}g(s,\xi)}{\Phi_s(u,s)+c}\right) \, \partial_s e^{i\Phi(u,s)+ics}ds\\
=-i\frac{g(t,\xi)}{\Phi_t(u,t)+c}+ \frac{-ic+\partial_tg}{(\Phi_t(u,t)+c)^2}- \frac{g\Phi_{tt}}{(\Phi_t(u,t)+c)^3} -i e^{-i\Phi(u,t)} \int_0^t \,   \partial_s\left[\frac1{\Phi_s(u,s)+c}\partial_s\left( \frac{e^{-ics}g(s,\xi)}{\Phi_s(u,s)+c}\right)\right] \, e^{i\Phi(u,s)+ics}ds\\
:=-i\frac{g(t,\xi)}{\Phi_t(u,t)+c}+g_1(t,\xi)\ \ \ \text{where\ } g_1(t,\xi)=O(\xi^{-3})\ (\xi\to\pm\infty)
\end{multline}

Similarly, in \eqref{hatpsimin} $\hat{\psi}_-$ is a sum of two terms which are, up to multiplicative constants, $T_{1,-}$ and $T_{2,-}$ which are obtained from $T_1,T_2$ above by replacing $g$ with $-g$ and for $\Phi(u,t)$ by $\xi^2t/2$. It follows that we have
$T_{2,-}+T_2=O(\xi^{-3})$.

Integrating by parts twice in $T_1$ and using the fact that $f(0)= f'(0)=0$ we obtain
$$T_1=e^{-i\Phi(u,t)-i\frac E{\omega^2}\xi}\,\int_0^\infty e^{-iy\xi}f(y)\,dy\big|_{u=\xi-\frac E\omega \sin\omega t}=\frac{-1}{\xi^2} e^{-i\Phi(u,t)-i\frac E{\omega^2}\xi}\,\int_0^\infty e^{-iy\xi}f''(y)\,dy$$
and similarly for $T_{1,-}$.

Now $\hat{\psi}=\hat{\psi}_-+\hat{\psi}_+$ and we see that $\partial_t\hat{\psi},\xi\hat{\psi},\xi^2\hat{\psi}$ are in $L^2(\RR,d\xi)$ for each $t$.  Returning to Eq. \eqref{schrodinger2}, we see that, for any $t$, $\psi_t$ and  $\partial_x^2 \psi$ are in $L^2$,implying straightforwardly that $x\psi\in L^2$.  Writing now, as usual,  the equation for $\frac{d}{dt}\|\psi\|(\cdot,t)\|_2^2$ it follows that $\|\psi\|(\cdot,t)\|_2$ is conserved. Since the evolution is reversible, it is unitary.

\section{Proof of Theorem\,\ref{theo:existence}}\label{S:ProofT2}

The proof relies on the following Lemmas, proved below.

  \subsection{The equation for $\psi(0,t)$} \label{Rigres}

  Let $\mathcal{D}$ be as defined in \eqref{defD}.
  We use the convolution
 \begin{equation}(f*g)(s)=\int_0^s f(u)g(s-u)\, du
 \label{def_convolution}\end{equation}

\begin{Lemma}\label{lemma:matching} Assume the initial condition $f$ satisfies $f\in \mathcal{D}$.

Let $\psi_-(0,t)$ be given by \eqref{psim0} and $\psi_+(0,t)$ given by \eqref{psiplus0}.

We have $\psi_-(0,t)=\psi_+(0,t)=\psi_0$ if and only if \eqref{eq:relpsipsix} holds and $\psi_0$
 satisfies the integral equation
\begin{equation}
  \psi_0(t)=h(t)+L\psi_0(t)
  \label{eqcontr}
\end{equation}  with 
  \begin{equation}
    h(t)= h_+(0,t)+ h_-(0,t) -\frac{1}{\pi}\int_0^t(h_-*s^{-1/2})\,  G(s,t)\,ds
    \label{h}
  \end{equation}
  (see (\ref{def_convolution}), in which $s^{-1/2}$ stands for the function $s\mapsto s^{-1/2}$) and
  \begin{equation}
    \begin{array}{r@{\ }>\displaystyle l}
      L\psi_0(t):=
      &\frac{1}{2\pi}\int_0^t(\psi_0*s^{-1/2}) G(s,t)\, ds\\[0.5cm]
      &+\frac{E}{2\omega\sqrt{2i\pi}}\int_0^t ds \, \,\psi_0(s)\, \frac 1{\sqrt{t-s}}\,  \left( \sin(\omega s)\, + \frac{ \cos(\omega t)- \cos(\omega s)}{\omega(t-s)}\right) \,e^{iF_0(s,t)}
    \end{array}
    \label{L}
  \end{equation}
  Here
  \begin{equation}
    G(s,t)=\frac d{ds}\left[\frac{e^{iF_0(s,t)}-1}{ \sqrt{t-s}}\right]
    \label{Gdef}
  \end{equation}
  and $F_0$ is given by \eqref{Fzero}.
  Furthermore,
  \begin{equation}\label{fomulapsix0}
  \psi_{x,0}= \frac{{\sqrt{2}}}{\sqrt{i \pi}} \frac{d}{dt}\left[ \psi_0(t)*t^{-1/2}- 2 h_-(0,t)*t^{-1/2}\right]  
  \end{equation}
\end{Lemma}
The proof is given in \S\ref{3p1}.

\begin{Lemma}\label{lemma:contraction}

Consider equation  \eqref{eqcontr}  with $h$ given by \eqref{h} and $L$ by \eqref{L}.  Assume the initial condition $f$ satisfies $f\in \mathcal{D}$. 

(i)  There exists $\nu_0>0$ such that, if $\nu>\nu_0$, then \eqref{eqcontr} is a contraction in the Banach space
  \begin{equation}
    \mathcal B_\nu:=\{\psi_0(t):\ e^{-\nu t}\psi_0(t)\in L^\infty(\mathbb R_+)\}
    .
    \label{Bnu}
  \end{equation}
  
  (ii) The functions $h_-$ and $h_+$  defined in \eqref{hminus0} and \eqref{hplus} resp. are differentiable for $t>0$.

  (iii) The solution $\psi_0$ of \eqref{eqcontr}, unique in $\mathcal B_\nu$, is continuously differentiable.
  
  (iv) Moreover, $\psi'_0:=\tfrac{d}{dx}|_{x=0}\psi$ is H\"older continuous of exponent $1/4$.

\end{Lemma}

{\bf Remark.} If $f$ is of class $C^r$ then $h_\pm$ are of class $C^r$.

The proof of Lemma\,\ref{lemma:contraction} is found in \S\ref{PfL6}.

  \begin{Lemma}\label{lemma:psipm} 
  Assume the initial condition $f$ satisfies $f\in \mathcal{D}$ and let $\psi_0,\psi_{x,0}$ be given by Lemma\,\ref{lemma:contraction}. 
  
  (i) The function $\psi_-$ given by \eqref{psiminus} is a solution of \eqref{schrodinger2} for $x<0$ and satisfies $\psi_-(x,0)=f(x)$ for $x<0$, $\psi_-(0-,t)=\psi_0(t)$, $\partial_x\psi_-(0-,t)=\psi_{x,0}(t)$.
  
  (ii) The function $\psi_+$ given by \eqref{psiplus} is a solution of \eqref{schrodinger2} for $x>0$ and satisfies $\psi_+(x,0)=f(x)$ for $x>0$, $\psi_+(0+,t)=\psi_0(t)$, $\partial_x\psi_+(0+,t)=\psi_{x,0}(t)$.
  
  (iii) The Fourier transform of $\psi_-$ is \eqref{hatpsimin} and the Fourier transform of $\psi_+$ is \eqref{hatpsiplu}. $\psi_-(\cdot,t)$ and $\psi_+(\cdot,t) $are $L^2$ functions.

  \end{Lemma}

The proof is found in \S\ref{PfLema4}.
 
\begin{Note}\label{lemma:x}
If $f\in L^\infty(\RR)$ formulas \eqref{psiminus}--\eqref{hplus} also hold in the sense of distributions. This is needed in order to accommodate initial conditions of the form \eqref{init}.
\end{Note}

\subsection{Proof of Lemma\,\ref{lemma:matching}} \label{3p1}

Relation \eqref{eq:relpsipsix} is precisely the condition that $\psi_-(0,t)=\psi_0$. We will now use this to eliminate $\psi_{x,0}$ from the condition $\psi_+(0,t)=\psi_0$.

Equation \eqref{eq:relpsipsix} implies
\begin{equation}\label{eq:psiprime}
\psi_{x,0}*t^{-1/2}=\frac{\sqrt{2\pi}}{\sqrt{i}}\left[ \psi_0(t)- 2 h_-(0,t)\right]
\end{equation}
which convolved with $t^{-1/2}$, and using the fact that $t^{-1/2}*t^{-1/2}=\pi$, gives
\begin{equation}
  \label{eq:postconvo}
   \int_0^{t} \psi_{x,0}(s)\,ds=\frac{\sqrt{2}}{\sqrt{i \pi}}\left[ \psi_0(t)*t^{-1/2}- 2 h_-(0,t)*t^{-1/2}\right]   
 \end{equation}
Note that this also proves\-~\eqref{fomulapsix0}.

The condition that $\psi_+(0+,t)=\psi_0(t)$ is equivalent to
\begin{equation}\label{eqpsip}
\psi_0(t)=2h_+(0,t)+ \mathcal{T}+ \frac1{{\sqrt{2\pi i}}}  \int_0^{t} \, ds\, \psi_0(s)\,g(s,t)\,{\rm e}^{iF_0(s,t)}
\end{equation}
where
$$\mathcal{T}= \frac1{{\sqrt{2\pi i}}}  \int_0^{t} \, \left[ -i \psi_{x,0}(s)\right]\,  \frac 1{\sqrt{t-s}}\,{\rm e}^{iF_0(s,t)} \,ds$$
and 
$$g(s,t)=\frac E\omega\frac{\sin \omega t}{{\sqrt{t-s}}}+\frac E{\omega^2} \frac{ \cos \omega t- \cos \omega s}{(t-s)^{3/2}}$$

 Noting that $e^{iF_0(s,t)}=1+(t-s)\Psi(s,t-s)$ where $\Psi(s,z)$ is entire, and using \eqref{eq:psiprime}, integrating by parts, then using \eqref{eq:postconvo}, we rewrite $\mathcal{T}$ as
 \begin{multline} \label{eq:oneterm}
 \mathcal{T}=\frac{-i}{{2\sqrt{i\pi}}}   \int_0^{t}  \,  \psi_{x,0}(s)\frac 1{\sqrt{t-s}}ds+\frac{-i}{{2\sqrt{i\pi}}}   \int_0^{t}  \psi_{x,0}(s)\sqrt{t-s}\,\frac{e^{iF_0(s,t)}-1}{t-s}\,ds\\
 =-\psi_0(t)+2h_-(0,t)+\frac{-i}{{2\sqrt{i\pi}}}   \int_0^{t} \left[ \frac{d}{ds}\int_0^s\psi_{x',0}(u)\, du\right]\, \frac{e^{iF_0(s,t)}-1}{\sqrt{t-s}}\, ds\\
   =-\psi_0(t)+2h_-(0,t)+\frac{i}{{\sqrt{2\pi i}}}   \int_0^{t}  ds \frac d{ds}\left[ \frac{e^{iF_0(s,t)}-1}{\sqrt{t-s}}  \right]\int_0^s\psi_{x,0}(u)du\\
   =-\psi_0(t)+2h_-(0,t)+\frac{1}{\pi}\int_0^t(\psi_0*s^{-1/2}) \ G(s,t)\, ds
   -\frac{2}{\pi}\int_0^t(h_-*s^{-1/2}) \,G(s,t)]\,ds
 \end{multline}
  Substituting \eqref{eq:oneterm} in \eqref{eqpsip},  we obtain \eqref{eqcontr}.
  
  \bigskip
    
  \subsection{Proof of Lemma\,\ref{lemma:contraction}}\label{PfL6}
  {\bf (i)}
  We prove that \eqref{eqcontr} is a contraction in the Banach space $\mathcal B_\nu$ \eqref{Bnu}.
\bigskip

\indent
Defining $\|g\|_\nu:=\|g(s)e^{-\nu s}\|_\infty$,
We bound
\begin{equation}
  e^{-\nu t}\left|\int_0^{t} (\psi_0\ast s^{-\frac12})G(s,t)\ ds\right|
  \leqslant
  \|\psi_0\|_\nu\int_0^{t}  |G(s,t)|e^{-\nu(t-s)}\int_0^s\ \frac{e^{-\nu(s-u)}}{\sqrt{s-u}}\ duds
  .
\end{equation}
Furthermore,
\begin{equation}
  \int_0^s\ \frac{e^{-\nu(s-u)}}{\sqrt{s-u}}\ du
  =\frac{\sqrt\pi\mathrm{erf}(\sqrt{\nu s})}{\sqrt\nu}\leqslant\frac{\sqrt\pi}{\sqrt\nu}
  .
\end{equation}
Now, changing variables,
\begin{equation}
  \int_0^{t} |G(s,t)|e^{-\nu(t-s)}\ ds
  =
  \int_0^{t} |G(t-s,t)|e^{-\nu s}\ ds
  .
\end{equation}
We then write
\begin{equation}
  G(t-s,t)=\frac{e^{iF_0(t-s,t)}-1-is\partial_s F_0(t-s,t)}{s^{\frac32}}
\end{equation}
and by \eqref{Fzero}, as $s\to 0$, $F_0(t-s,t)\sim {\rm const.}s$, so $G(t-s,t)\to0$ as $s\to 0$.
On the other hand, for large $s$, $\partial_sF_0(t-s,t)$ is bounded, so $G(s,t)$ is as well.
Thus
\begin{equation}
  \int_0^{t} |G(t-s,t)|e^{-\nu s}\ ds
  =
  O(\nu^{-1})
\end{equation}
and
\begin{equation}
  \left\|\int_0^{t} (\psi_0\ast s^{-\frac12})G(s,t)\ ds\right\|=O(\nu^{-\frac32})\|\psi_0\|_\nu
  .
  \label{contraction1}
\end{equation}
\bigskip

\indent
Similarly, 
\begin{multline}
  e^{-\nu t}\left|\int_0^{t} \psi_0(s)\, \frac 1{\sqrt{t-s}}\,  \left( \sin(\omega s)\, + \frac{ \cos(\omega  t)- \cos(\omega s)}{\omega (t-s)}\right) \,{\rm e}^{iF_0(s,t)}\ ds\right| \\
  \leqslant\|\psi_0\|_\nu\int_0^{t} \ e^{-\nu(t-s)}\frac{\sin(\omega s)+\frac{\cos(\omega  t)-\cos(\omega s)}{\omega (t-s)}}{\sqrt{t-s}}\ ds
\end{multline}
in which we change variables:
\begin{equation}
  \int_0^{t} \ e^{-\nu(t-s)}\frac{\sin(\omega s)+\frac{\cos(\omega t)-\cos(\omega s)}{\omega (t-s)}}{\sqrt{t-s}}\ ds
  =
  \int_0^{t} \ e^{-\nu s}\frac{\sin(\omega t-\omega s)+\frac{\cos(\omega t)-\cos(\omega t-\omega s)}{s}}{\sqrt{s}}\ ds
\end{equation}
and since, as $s\to0$, $\sin(\omega t-\omega s)+\frac{\cos(\omega t)-\cos(\omega t-\omega s)}{s}\sim {\rm const.}s$, so the integrand is bounded as $s\to0$.
For large $s$, the integrand is obviously bounded above, so
\begin{equation}
  \left\|\int_0^{t} \psi_0(s)\, \frac 1{\sqrt{t-s}}\,  \left( \sin(\omega s)\, + \frac{ \cos(\omega t)- \cos(\omega s)}{\omega (t-s)}\right) \,{\rm e}^{iF_0(s,t)}\ ds\right\|_\nu
  =O(\nu^{-1})\|\psi_0\|_\nu
\end{equation}
Combining this with \eqref{contraction1}, we find that
\begin{equation}
  \|L\psi_0\|_\nu=O(\nu^{-1})\|\psi_0\|_\nu
  .
\end{equation}
Therefore, for $\nu$ large enough, $\mathds 1-L$ is invertible in $\mathcal B_\nu$, so \eqref{eqcontr} is a contraction.

\

{\bf  (ii)}
To prove that $h_-(0,t)$ is differentiable we split the integral in \eqref{hminus} into the integral from $-1$ to $0$, which is clearly differentiable plus the integral from $-\infty$ to $-1$, which we show it is differentiable using 
$L^2$ limits to integrate by parts as follows. We have
$$l.i.m\, \int_{-\infty}^{-1} dy \, f(y) \, {\rm e}^{\frac{iy^2}{2t}}=\frac12\, l.i.m\, \int_1^{\infty} du \, \frac{f(-\sqrt{u})}{\sqrt{u}} \, {\rm e}^{\frac{iu}{2t}}$$
and, integrating by parts twice we find\footnote{The boundary terms vanish since we are dealing with an $L^2$ function which is continuous in $R$, cf. Note \ref{Note}, hence it goes to zero along some subsequence $\{R_n\}_{n\in\NN}$ where $\lim_{n\to \infty}R_n=\infty$.}
$$
  \frac12\, l.i.m\, \int_1^{\infty} du \, \frac{f(-\sqrt{u})}{\sqrt{u}} \, {\rm e}^{\frac{iu}{2t}}
  =
  -itf(-1) \, {\rm e}^{\frac{i}{2t}}
  -t^2{\rm e}^{\frac{i}{2t}}(f(-1)+f'(-1))
  -2t^2\, l.i.m\, \int_1^{\infty} du \, {\rm e}^{\frac{iu}{2t}} \, \frac{d^2}{du^2}\frac{f(-\sqrt{u})}{\sqrt{u}}
$$
The first two terms are obviously differntiable for $t\in(0,\infty)$, so it suffices to consider the integral term.
The second derivative above is a sum of terms of the form: $\frac{f''(-\sqrt{u})}{u^{3/2}} $, $\frac{f'(-\sqrt{u})}{u^{2}} $, $\frac{f(-\sqrt{u})}{u^{5/2}} $.

Since $f''\in L^2$ then the following quantities are finite:
$$\int_{-\infty}^{-1} |f''(y)|^2\,dy=\int_1^{\infty} |f''(-\sqrt{u})|^2\, \frac1{2\sqrt{u}}\, du<\infty$$ 
hence $\frac{f''(-\sqrt{u})}{u^{1/4}}$ is in $L^2$. The other two terms,  $\frac{f'(-\sqrt{u})}{u^{2}} $, $\frac{f(-\sqrt{u})}{u^{5/2}} $ have faster decay. Then $\frac{d^2}{du^2}\frac{f(-\sqrt{u})}{\sqrt{u}} ={g_3}(u)u^{-5/4} $ with ${g_3}$ in $L^2$. Denoting $\tau=1/(2t)$, we need to show that ${G_3}(\tau):=\int e^{iu\tau} {g_3}(u)u^{-5/4}\, du$ is differentiable in $\tau$. Calculate then
$${G_3}(\tau+\epsilon)-{G_3}(\tau)=\int_1^\infty e^{iu\tau} {g_3}(u)u^{-5/4}\, \left(e^{iu\epsilon}-1\right) du$$
We have 
$e^{ix }=1+ix+g_1(x)\, x^{5/4}$
where $g_1:=(e^{ix}-1-ix)x^{-\frac54}$ is a continuous, bounded function.  Therefore
$$\frac{{G_3}(\tau+\epsilon)-{G_3}(\tau)}{i\epsilon}=\int_1^\infty e^{iu\tau} {g_3}(u)u^{-1/4}\, du+I_\epsilon,\ \ \text{where}\ I_\epsilon=i^{\frac14}\epsilon^{1/4}\int_1^\infty e^{iu\tau} {g_3}(u) g_1(u\epsilon)\, du$$
Using the fact that the integral in $I_\epsilon$ is the Fourier transform of the $L^2$ function ${g_3}g_1\chi_{[1,\infty)}$, then its  $L^2$ norm is bounded by $\epsilon^{1/4}\|{g_3}\| \sup |g_1|$ hence $I_\epsilon$ goes to $0$ in the $L^2$ norm,  hence in $L^1_{\rm loc}$. It follows that ${G_3}$ is differentiable in distributions and its derivative is $i\int_1^\infty e^{iu\tau} {g_3}(u)u^{-1/4}\, du$, an $L^2$ function (hence  $L^1_{\rm loc}$) implying that ${G_3}$ is absolutely continuous, hence differentiable a.e.

Now it follows that ${G_3}'$ is continuous a.e. since, using $e^{ix }=1+g_2(x)\, x^{1/4}$
where $g_2$ is a continuous, bounded function, we have
\begin{equation}\label{eq:g3Holder}
  {G_3}'(\tau+\epsilon)-{G_3}'(\tau)=i\int_1^\infty e^{iu\tau} {g_3}(u)u^{-1/4}\left(e^{iu\epsilon}-1\right)\, du =(i\epsilon)^{1/4}i\int_1^\infty e^{iu\tau} {g_3}(u)\, du
\end{equation}
which goes to zero as $\epsilon\to 0$, as $I_\epsilon$ did before.
We have also shown:
\begin{Lemma}\label{L10-a}
  The function $G_3$ is differentiable and the derivative is H\"older continuous of exponent $1/4$ uniformly in $\tau$.
\end{Lemma}
Indeed the integral in the last term of \eqref{eq:g3Holder} is bounded.

Clearly,  $h_+(0,t)$ is differentiable if and only if 
$\int_0^{\infty }dy\, f(y)\, e^{i\frac{[y+\frac{E}{\omega^2}(1-\cos(\omega t))]^2}{2t}  }$ is differentiable. Let $u=[y+\frac{E}{\omega^2}(1-\cos(\omega t))]^2$. We need to show differentiability in $\tau$ of $\int_{\frac{E^2}{\omega^4}(1-\cos(\omega t))^2}^{\infty }dy\, f(\sqrt{u}-\frac{E}{\omega^2}(1-\cos(\omega t))) u^{-1/2}\, e^{i\tau u} $, for which it suffices to show that  $\int_c^{\infty }dy\, f(\sqrt{u}-\frac{E}{\omega^2}(1-\cos(\omega t))) u^{-1/2}\, e^{i\tau u} $ is differentiable, where $c$ is a constant large enough. The rest of the proof is similar to the one above for $h_-(0,t)$.

\

{\bf (iii)}
To prove regularity of  $\psi_0$, note first that, since $\psi_0\in\mathcal B_\nu\subset L^\infty_{loc}$ then $L\psi_0$ is continuous, since integrals of the form $\int_0^t\psi_0(s)(t-s)^{-1/2}f(s,t)ds$ with $\psi_0\in L^{\infty}_{\rm loc}$ and $f$ continuous are continuous in $t$. Therefore, since $h(t)$ is differentiable, $\psi_0$ is continuous. 

Then, iterating \eqref{eqcontr}, it follows that $L\psi_0$ is differentiable, as follows. We have $\psi_0=L\psi_0+h$ where $h$ is differentiable and $L\psi_0$ has the form
$L\psi_0(t)=\int_0^t ds\, \psi_0(s)(g_1(s,t)+\frac{1}{\sqrt{t-s}}g_2(s,t))$ and we will now show that $g_1,g_2$ are analytic in $s,t$.
By\-~\eqref{L},
\begin{equation}
  g_1(s,t)=\frac1{2\pi}\int_s^t du\ \frac{G(u,t)}{\sqrt{u-s}}
  ,\quad
  g_2(s,t)=
  \frac{E}{2\omega\sqrt{2i\pi}}\left( \sin(\omega s)\, + \frac{ \cos(\omega t)- \cos(\omega s)}{\omega(t-s)}\right) \,e^{iF_0(s,t)}
  .
\end{equation}
By\-~\eqref{Fzero}, $F_0$ is analytic, and so $g_2$ is as well.
As for $g_1$, we rewrite\-~\eqref{Gdef} as
\begin{equation}
  G(u,t)
  =\frac{G_1(u,t)}{\sqrt{t-u}}
  ,\quad
  G_1(u,t):=
  i\partial_u F_0(u,t)e^{iF_0(u,t)}-\frac{e^{iF_0(u,t)}-1}{2(t-u)}
\end{equation}
in terms of which
\begin{equation}
  g_1(s,t)=\frac1{2\pi}\int_s^t\ du\ \frac{G_1(u,t)}{\sqrt{t-u}\sqrt{u-s}}
  .
\end{equation}
Furthermore, by\-~\eqref{Fzero}, $G_1$ is analytic in $u,t$, and since for $n\geqslant 0$,
\begin{equation}
  \frac1{2\pi}\int_s^t\ du\frac{(t-u)^n}{\sqrt{t-u}\sqrt{u-s}}
  =
  \frac{\Gamma(n+\frac12)}{2\sqrt\pi n!}(t-s)^n
\end{equation}
we have
\begin{equation}
  g_1(s,t)=\sum_{n=0}^\infty\frac1{n!}\frac{\Gamma(n+\frac12)}{2\sqrt\pi n!}(t-s)^n\left.\frac{\partial^n G_1(u,t)}{\partial u^n}\right|_{u=t}
\end{equation}
which is analytic in $s,t$.
We then split $\psi_0(t)=\int_0^t ds\,\psi_0(s)(g_1(s,s)+\frac1{\sqrt{t-s}}g_2(s,s))+h_1(t)$ with $h_1:=h(t)+\int_0^tds\psi_0(s)((g_1(s,t)-g_1(s,s))+\frac1{\sqrt{t-s}}(g_2(s,t)-g_2(s,s))$ which is differentiable.
We now iterate this formula:
\begin{equation}
  \psi_0(t)=\int_0^t ds\,\left(g_1(s,s)+\frac{1}{\sqrt{t-s}}g_2(s,s)\right)\left[ \int_0^s d\sigma\,\psi_0(\sigma)\left(g_1(\sigma,\sigma)+\frac1{\sqrt{t-\sigma}}g_2(\sigma,\sigma)\right)+h_1(s)\right]     +h_1(t)
\end{equation}
in which we change the order of integration to find
\begin{equation}
  \psi_0(t)=\int_0^t d\sigma\,\psi_0(\sigma)\left[ \int_\sigma^t ds\,\left(g_1(s,s)+\frac{1}{\sqrt{t-s}}g_2(s,s)\right)\left(g_1(\sigma,\sigma)+\frac1{\sqrt{t-\sigma}}g_2(\sigma,\sigma)\right)+h_1(s)\right]     +h_1(t)
  .
\end{equation}
In this integral, $g_1,g_2$ are analytic in a neighborhood of $\RR^+$ and $\psi_0$ is continuous, hence the integral is differentiable with continuous derivative.
Using Lemma \ref{L10-a}, the same arguments, and the fact that the integral operators preserve H\"older continuity,  show (iv) holds.

\subsection{Proof of Lemma\,\ref{lemma:psipm} }\label{PfLema4}
We will first prove (iii), and then move on to (i) and (ii).
\bigskip

{\bf  (iii)}
 {\em For $x<0$} we show that the function given by \eqref{hatpsimin} is in $L^2$, we take its  inverse Fourier transform and show that the result is \eqref{psiminus} which is an $L^2$ function.
 
For the first term in  \eqref{hatpsimin}, note that since $f$ is in $L^2(\RR)$ then by \eqref{Cminus}, so is $C_-$ hence so is the inverse Fourier transform of ${\rm e}^{-i\xi^2 t/2} C_-(\xi)$. We have (see Note\,\ref{Note})
\begin{multline}\label{calchmin}
 l.i.m.\, \frac 1{{2\pi}}\int_{-N}^N d\xi\,  {\rm e}^{ix\xi -i\xi^2 t/2} l.i.m.\,\int_{-N}^0 dy\,  {\rm e}^{-iy\xi }f(y)\\
=l.i.m. \frac 1{{2\pi}}\int_{-N}^0 dy \, f(y) \int_{-N}^N d\xi \,{\rm e}^{ix\xi -i\xi^2 t/2-iy\xi} =l.i.m.\,\frac 1{{2\pi}}\int_{-N}^0 dy \, f(y)\,\frac{\sqrt{2\pi}}{\sqrt{it}} {\rm e}^{\frac{i(x-y)^2}{2t}}
\end{multline}
yielding \eqref{hminus}, and that $h_-(\cdot, t)$ is an $L^2$ function.

The second term in  \eqref{hatpsimin} is $\tfrac 1{2\sqrt{2\pi}}e^{-i\xi^2 t/2}\int_0^{t} ds\, {\rm e}^{i\xi^2 s/2}\   i\psi_{x,0}(s)$, is also an $L^2$ function. Indeed, from \eqref{fomulapsix0} we have $\psi_{x,0}(t)=u*t^{-1/2}$ with $u={\rm const}(\psi_0'-2\partial_th_-(0,\cdot))$ hence, after changing the order of integration and a substitution we have
$$\int_0^{t} ds\, {\rm e}^{i\xi^2 s/2}\  \psi_{x,0}(s)=\int_0^td\sigma\, u(\sigma)\,  {\rm e}^{i\xi^2 \sigma/2}\int_0^{t-\sigma}d\tau\,  {\rm e}^{i\xi^2 \tau/2}\tau^{-1/2} $$
and the last integral can be explicitly calculated and, for large $|\xi|$, it is less than const. $|\xi|^{-1}$.

The Fourier transform of this second term can be then computed as was done above for $h_-$, yielding 
\begin{equation}\label{term2minus}
\frac {\sqrt{i}}{2\sqrt{2\pi}}\int_0^t\, ds\, \psi_{x,0}(s)\frac 1{\sqrt{t-s}}e^{\frac{ix^2}{2(t-s)}}
\end{equation}

The third term in \eqref{hatpsimin} is also in $L^2$, since integrating by parts we have
\begin{multline}\label{modelxneg}
\int_0^{t} ds\, {\rm e}^{i\xi^2 s/2}\   \xi \psi_{0}(s)=\xi \int_0^{t} ds\, {\rm e}^{i(\xi^2+1) s/2}\  {\rm e}^{-i s/2} \psi_{0}(s)\\
=\frac{-2i\xi}{\xi^2+1} \left[ {\rm e}^{i \xi^2 t/2} \psi_0(t)-\psi_{0}(0)+\int_0^{t} ds\,{\rm e}^{i\xi^2 s/2}\  \left(\frac i2 \psi_{0}(s) -\psi_0'(s)\right)\, ds  \right]
\end{multline}
which,  since $\psi$ and $\psi'$ are locally bounded by lemma \ref{lemma:contraction}, is manifestly in $L^2$. To calculate its inverse Fourier transform we write
\begin{multline}\label{I4minus}
\frac{-1}{4\pi}\, \underset{\epsilon\to 0}{\underset{N\to\infty}{l.i.m.}}\, \int_{-N}^{N}\, d\xi\, {\rm e}^{ix\xi}{\rm e}^{-i\xi^2 t/2}\xi\int_0^{t-\epsilon}\, ds\, \psi_0(s){\rm e}^{i\xi^2 s/2}=\frac{-1}{4\pi}\, l.i.m.\,\int_0^{t-\epsilon}\, ds\, \psi_0(s)\, \int_{-N}^{N}\, d\xi\, {\rm e}^{ix\xi}{\rm e}^{-i\xi^2 (t-s)/2}\xi\\
=\frac{i}{4\pi} \,\underset{\epsilon\to 0}{l.i.m.}\,\int_0^{t-\epsilon}\, ds\, \psi_0(s)\,\partial_x \int_{-\infty}^{\infty}\, d\xi\, {\rm e}^{ix\xi}{\rm e}^{-i\xi^2 (t-s)/2}=
\frac{i}{4\pi} \int_0^t\, ds\, \psi_0(s)\,\partial_x \frac{\sqrt{2\pi}}{\sqrt{i}}\frac1{\sqrt{t-s}}e^{\frac{ix^2}{2(t-s)}}
\end{multline}

Adding up \eqref{calchmin}, \eqref{term2minus} and \eqref{I4minus} we obtain \eqref{psiminus}.

\

{\em For $x>0$,} we show that taking inverse Fourier transform in \eqref{hatpsiplu}-\eqref{Cplus} we obtain \eqref{psiplus}, an $L^2$ function.

The inverse Fourier transform in \eqref{hatpsiplu} is a sum of two tems: $I_1+I_2$ where
$I_1$ is the inverse Fourier transform of $e^{-i\Phi(u,t)}C_+(u)$ (where $u=\xi-\frac E\omega\sin(\omega t)$):
$$I_1=\frac{1}{2\pi}\int_{-\infty}^{\infty}\, d\xi\, e^{ix\xi}e^{-i\Phi\left(\xi-\frac E\omega\sin(\omega t),t\right)-i\frac{E}{\omega^2}\xi}\int_0^\infty\, dy\, e^{-iy\xi}f(y) =h_+(x,t)$$
where the calculation is similar to that of $h_-$, and yields \eqref{hplus} (which is an $L^2$ function since $f$ is).

\begin{multline}\label{I2plus}
 I_2=l.i.m.\frac1{4\pi}\int_{-N}^Nd\xi e^{ix\xi}e^{-i\Phi(u,t)}\int_0^t \, e^{i\Phi(u,s)}   [ -i \psi_{x,0}(s)+\xi\psi_0(s)]\,ds \\ 
=\frac1{4\pi}e^{ix \frac E\omega \sin\omega t}l.i.m.\int_{-N}^Ndu\, e^{ixu}e^{-i\Phi(u,t)}\int_0^t \, e^{i\Phi(u,s)}  [ -i \psi_{x,0}(s)+(u+\frac E\omega \sin\omega t)\psi_0(s) ]\,ds \\ 
=\frac1{4\pi}e^{ix \frac E\omega \sin\omega t}\,\int_0^t \,ds  [ -i \psi_{x,0}(s)+\frac E\omega \sin\omega t\, \psi_0(s)]\, \ l.i.m.\,\int_{-N}^Ndu\, e^{ixu-i\Phi(u,t)+i\Phi(u,s)} \,ds \\
+ \frac1{4\pi}e^{ix \frac E\omega \sin\omega t}\,\int_0^t \,ds\ \psi_0(s)\,  l.i.m.\,\int_{-N}^Ndu\, \  u\ e^{ixu-i\Phi(u,t)+i\Phi(u,s)} \,ds \ =:I_3+I_4
\end{multline}

To continue the calculations in \eqref{I2plus}, we have
\begin{multline}\label{I3plus}
I_3=\frac1{4\pi}e^{ix \frac E\omega \sin\omega t -i(U+\frac{E^2}{4\omega^2})t}\\
\times\  \int_0^t \,ds  [ -i \psi_{x,0}(s)+\frac E\omega \sin\omega t\, \psi_0(s)]\ 
 e^{-i\frac{E^2}{8\omega^3}[\sin(2\omega t)-\sin(2\omega s)]+i(U+\frac{E^2}{4\omega^2})s}   \ \frac{\sqrt{2\pi}}{\sqrt{i}\sqrt{t-s}}\, e^{i\frac{[x+\frac{E}{\omega^2} [\cos\omega t -\cos\omega s]]^2}{2(t-s)}}
  \end{multline}
  Thus
  \begin{equation}\label{I3final}
    I_3 =\frac 1{2\sqrt{2\pi i}}\  \int_0^t \,ds\,  \frac{ -i \psi_{x,0}(s)+\frac E\omega \sin\omega t\, \psi_0(s)}{\sqrt{t-s}}\   \, e^{iF(x,s,t)}
 \end{equation}
 where $F$ is given by \eqref{defFxst}.

We evaluate $I_4$ in a way similar to \eqref{I4minus}:
\begin{multline}\label{I4plus1}
I_{4}= \frac1{4\pi}e^{ix \frac E\omega \sin\omega t}\ l.i.m.\int_0^{t-\epsilon} \,ds\,\psi_0(s)\, \int_{-N}^Ndu\, \  u\ e^{ixu-i\Phi(u,t)+i\Phi(u,s)} \,ds \\
 =\frac1{4\pi}e^{ix \frac E\omega \sin\omega t}\ l.i.m.\int_0^{t-\epsilon} \,ds\,\psi_0(s)\, (-i)\partial_x\,\int_{-N}^Ndu\,  e^{ixu-i\Phi(u,t)+i\Phi(u,s)} \,ds \\
 =\frac1{4\pi}e^{ix \frac E\omega \sin\omega t   -i(U+\frac{E^2}{4\omega^2})t  }\  \underset{\epsilon\to 0}{l.i.m}\,\int_0^{t-\epsilon} \,ds\,\psi_0(s)\, (-i)
 e^{-i\frac{E^2}{8\omega^3}[\sin(2\omega t)-\sin(2\omega s)]   +i(U+\frac{E^2}{4\omega^2})s  } \\
 \times\,  \ \frac{\sqrt{2\pi}}{\sqrt{i}\sqrt{t-s}}\, \partial_x\, e^{i\frac{[x+\frac{E}{\omega^2} [\cos(\omega t) -\cos(\omega s)]]^2}{2(t-s)}}\\
 \end{multline}
 \begin{multline}\label{I4plus}
  =\frac1{2\sqrt{2\pi i}}e^{ix \frac E\omega \sin\omega t   -i(U+\frac{E^2}{4\omega^2})t }\ \int_0^t \,ds\,\psi_0(s)\, (-i)
 e^{-i\frac{E^2}{8\omega^3}[\sin(2\omega t)-\sin(2\omega s)]   +i(U+\frac{E^2}{4\omega^2})s}   \\
 \times\,  i\,\frac{\frac{E}{\omega^2} [\cos(\omega t) -\cos(\omega s)]+x}{(t-s)^{3/2}}  \, e^{i\frac{\left[\frac{E}{\omega^2} (\cos\omega t -\cos\omega s)+x\right]^2}{2(t-s)}}
\end{multline}
thus
\begin{equation}\label{I4final}
I_4=\frac1{2\sqrt{2\pi i}}\ \int_0^t \,ds\,\psi_0(s)\, \frac{\frac{E}{\omega^2} [\cos(\omega t) -\cos(\omega s)]+x}{(t-s)^{3/2}}  \, e^{iF(x,s,t)}
\end{equation}
 a convergent improper integral, with $F$ given by \eqref{defFxst}. 

Adding \eqref{I2plus}, \eqref{I3plus} and \eqref{I4plus} we obtain \eqref{psiplus}, an $L^2$ function.

\bigskip

{\bf (i)} 
The fact that \eqref{psiminus} is a solution of \eqref{schrodinger2} for $x<0$ is a simple calculation.

We will now calculate the limit of \eqref{psiminus} as $x\to0-$. 
Note that, for $x<0$,
\begin{equation}\label{beforewhere}
\int_0^{t} ds\, \psi_0(s) \,  \frac {ix}{(t-s)^{3/2}}{\rm e}^{\frac{ix^2}{2(t-s)}} = \psi_0(t)  \int_0^{t} ds\, \frac{ix}{(t-s)^{3/2}}{\rm e}^{\frac{ix^2}{2(t-s)}}
+
\int_0^{t} ds\, \frac{ix(\psi_0(s)-\psi_0(t))}{(t-s)^{3/2}}{\rm e}^{\frac{ix^2}{2(t-s)}}
\end{equation}
and
\begin{equation}
\int_0^{t} ds\, \frac{ix(\psi_0(s)-\psi_0(t))}{(t-s)^{3/2}}{\rm e}^{\frac{ix^2}{2(t-s)}}
  =O(x)
  .
  \label{diffpsi}
\end{equation}
Furthermore,
\begin{equation}
\psi_0(t)  \int_0^{t} ds\, \frac{ix}{(t-s)^{3/2}}{\rm e}^{\frac{ix^2}{2(t-s)}}
=- i \psi_0(t)\int_{x^2/t}^\infty \frac 1{\sqrt{\tau}}{ \rm e}^{i\tau/2}
\end{equation}
so
\begin{equation}
  \int_0^{t} ds\, \psi_0(s) \,  \frac {ix}{(t-s)^{3/2}}{\rm e}^{\frac{ix^2}{2(t-s)}}
  \to  -i\, \psi_0(t)\, \sqrt{2\pi i}\ \ (\text{as }x\to 0-)
  .
\end{equation}
Therefore, taking $x\to0-$,
$$\psi_-(0,t)=h_-(0,t)+ \frac {\sqrt{i}}{2\sqrt{2\pi}} \int_0^{t} ds\,\frac { \psi_{x,0}(s)}{\sqrt{t-s}} +\frac{1}{{2}}\psi_0(t)$$
and the right hand side in the above equals $\psi_0(t)$ by \eqref{eqcontr}.

The limit of \eqref{psiminus} as $t\to0+$ for $x<0$ equals $\lim_{t\to 0+}h_-(x,t)$. With the large parameter $t^{-1}$, the integrand has a saddle point at $y=x$, hence, by the saddle point method, equals $f(x)$.

\

{\bf  (ii)}
The fact that $\psi_+$ given by \eqref{psiplus} is a solution of \eqref{schrodinger2} for $x>0$ is a simple calculation.

We will now take the limit of \eqref{psiplus} as $x\to0+$. 
From \eqref{I3final} we have
\begin{equation}\label{I30}
\lim_{x\to 0+}I_3=\frac1{2\sqrt{2\pi i}}\ \int_0^t \,ds  [ -i \psi_{x,0}(s)+\frac E\omega \sin\omega t\, \psi_0(s)]\ 
  \frac1{\sqrt{t-s}}\, e^{iF(0,s,t)}
   \end{equation}

To calculate the limit of $I_4$,  we write $I_4=I_{41}+I_{42}$ where
$$I_{41}=\frac1{2\sqrt{2\pi i}}\ \int_0^t \,ds\,\psi_0(s)\, \frac{\frac{E}{\omega^2} [\cos(\omega t) -\cos(\omega s)]}{(t-s)^{3/2}}  \, e^{iF(x,s,t)},\ \ I_{42}=\frac1{2\sqrt{2\pi i}}\ \int_0^t \,ds\,\psi_0(s)\, \frac{x}{(t-s)^{3/2}}  \, e^{iF(x,s,t)}$$
We have
\begin{equation}\label{I410}
\lim_{x\to 0+}I_{41}=\frac1{2\sqrt{2\pi i}}\ \int_0^t \,ds\,\psi_0(s)\, \frac{\frac{E}{\omega^2} [\cos(\omega t) -\cos(\omega s)]}{(t-s)^{3/2}}  \, e^{iF(0,s,t)}
 \end{equation}
while (by the same reasoning as in \eqref{diffpsi},)
 \begin{equation}\label{I420}
 I_{42}=
 \frac1{2\sqrt{2\pi i}}\ \psi_0(t)\ \int_0^t \,ds\, \frac{x}{(t-s)^{3/2}}  \, e^{iF(x,s,t)}+O(x)
  \end{equation}
  Now, by \eqref{defFxst}, $\frac{F(x,s,t)-\frac{x^2}{2(t-s)}}{t-s}$ is analytic, hence from \eqref{I420} we further have
  \begin{equation}\label{I420}
  I_{42}=
 \frac1{2\sqrt{2\pi i}}\ \psi_0(t)\ \int_0^t \,ds\, \frac{x}{(t-s)^{3/2}}  \, e^{i \frac{x^2}{2(t-s)}}+O(x)\to \frac12 \psi_0(t)\ \ \text{as}\ x\to 0+
  \end{equation}
where the last limit is evaluated as in \eqref{beforewhere}.

Combining  \eqref{I30}, \eqref{I410}, \eqref{I420} we obtain \eqref{psiplus0}, whose right hand side equals $\psi_0(t)$, since $\psi_0$ satisfied the relations in Lemma\,\ref{lemma:matching}.

The limit of \eqref{psiplus} as $t\to0+$ equals
$$\lim_{t\to 0+}  h_+(x,t)= \lim_{t\to 0+}\    \frac1{\sqrt{2\pi i t }} \int_0^{\infty }dy\, f(y)\, e^{\frac{i\left[x-y\right]^2}{2t} } +O(t)=f(x)$$
where we used the saddle point method.

\subsection{Proof of Theorem\,\ref{theo:existence}}

Let $\psi_+$ and $\psi_-$ be given by Lemma\,\ref{lemma:psipm}. Then $\psi(x,t):=\psi_-(x,t)\Theta(-x)+\psi_+(x,t)\Theta(x)$ is an $L^2$ solution of \eqref{schrodinger2} with the initial condition $f$.

\section{Long time behavior: proof of Theorems\,\ref{Thlongtime} and \ref{theo:2}}\label{PfLongt}

Most of the technical elements of the proof of Theorem \ref{theo:2} are common with those of Theorem\,\ref{Thlongtime}. The only distinction (the presence of some additional poles due to the initial condition) are dealt with at the end of this section.

The discrete-Laplace transform technique used in this section was devised as an adaptation of Laplace-Borel methods used in \cite{CCe20}, \cite{CCL18}, in order to deal with the present setting of noncompact operators, see Appendix \ref{appen} for the connection between the two.

We perform a discrete Laplace transform (DLP) and the long time behavior of the system is now contained in the analytic properties of the transformed wave function with respect to the Laplace parameter $\sigma$. Namely, the discrete inverse Laplace transform (DILT), whose coefficients are obtained by Cauchy's formula, shows that the solution of the Schr\"odinger equation \eqref{schrodinger2} decay, as $t\to\infty$, uniformly for $x$ on compact sets, if and only if the Maclaurin coefficients decay with respect to their index $k$, which happens if and only if the DLP has no poles in the Laplace variable in an open neighborhood of the unit disk. The decay in $t$ mimics the decay of the coefficients with respect to their index $k$, and for the latter we show to have an upper bound of $k^{-1/2}$. 
The absence of poles is shown by proving the absence of discrete spectrum of the quasienergy operator where we use methods of Ecalle's theory of resurgence of transseries \cite{Co08,Ec81}.

The mathematical details in this section are as follows.
To avoid complicating the notations, in this section we assume $\omega=1$ (in fact $\omega$ can be rescaled in equation \eqref{schrodinger2}; see  appendix \ref{appen}  for $\omega$ not rescaled).

In \S\ref{discrete-Laplace transform} we define the DLT and its inverse DILT, and we show how it can be used for the study of integral equations of our type. In \S\ref{anstrsol} we study integral kernels with a singularity of the type we are dealing with, and give details on the techniques we use and results. In \S\ref{sovpereq} we discrete-Laplace transform the equation \eqref{eqcontr} for $\psi_0$ and deduce that its discrete-Laplace transform only has singularities of the square root branch point type and possible poles, with a finite number in any compact set, having the analytic structure \eqref{stareq20}. 
 In \S\ref{Proofofb} we show that existence of poles imply existence of nontrivial solutions of the quasienergy equation. The latter are ruled out in \S\ref{RAGE} based on Ecalle's theory of transseries, showing that the DLP of $\psi$ has no poles in a neighborhood of the closed unit disk. 

Combining all these elements the proof of Theorem\,\ref{Thlongtime} is completed in \S\ref{TheEnd}, and that of Theorem\,\ref{theo:2} is completed in \S\ref{lasttheo}.

\subsection{Discrete-Laplace Transform and long time behavior of $\psi(0,t)$}\label{discrete-Laplace transform}

The logic of the construction is as in \S\ref{Rigres}: we derive formally an integral equation for the discrete-Laplace transform (defined below) of $\psi_0(t):=\psi(0,t)$ , we show existence and uniqueness of solutions of that equation after which we check, in a straightforward way, that the solution is the discrete-Laplace transform of $\psi_0$.

 Let $\mathcal{S}$ be the space of  functions of the form $t\mapsto \Theta(t)g(t)$ which decay faster than $t^{-1-\epsilon}$.  For $f\in\mathcal{S}$ define its \emph{discrete-Laplace transform} for $\tau\in (-\pi,\pi]$ and $\sigma\in [0,1)$ by
 \begin{equation}\label{defP}
 (\mathcal{P}_\sigma f)(\tau):=\sum_{ k\ge 0}{\rm e}^{i\sigma{2k\pi}}f\left(\tau+{2k\pi}\right)
 \end{equation} 
 
Note that the function $f$ can be recovered from its discrete-Laplace transform by
\begin{equation}\label{invP}
 f\left(\tau+{2k\pi}\right) =\int_{0}^1\,d\sigma \, {\rm e}^{-i\sigma{2k\pi}}(\mathcal{P}_\sigma f)(\tau)
\end{equation} 
for all $\tau\in (-\pi,\pi]$ and $ k\ge 0$.

 For functions with not enough decay to ensure convergence of \eqref{defP} it is convenient to define a more general transform by taking complex $\sigma$ with $\Im\sigma>0$. Denote $z={\rm e}^{i\sigma{2\pi}}$ (note that $|z|<1$) and define  
   \begin{equation}\label{defPz}
 (\mathcal{P}_z f)(\tau):=\sum_{ k\ge 0} z^k f\left(\tau+{2k\pi}\right)
 \end{equation} 
 Then \eqref{defPz} is a generating function. Assume \eqref{defPz} converges in $D_\delta:=\{ z\in\CC\,|\, |z|<\delta\}$. Then the inversion relations \eqref{invP} are replaced by the Cauchy formula:
 \begin{equation}\label{invPz}
 f\left(\tau+{2k\pi}\right) =\frac 1{2\pi i}\oint_Cd\zeta\frac{ (\mathcal{P}_\zeta f)(\tau)}{\zeta^{k+1}}
\end{equation} 
where $C\subset D_\delta$ is a simple closed path around $0$.
 
 If the limit of $\mathcal P_z$ in\-~(\ref{defPz}) as $z$ approaches the unit circle exists, except possibly  at a discrete set of singular points (which we show to be our case), we take that limit (called the Abel sum, or Abel mean) as the discrete-Laplace transform of our function. As is well known, Abel summation of a convergent series is the ordinary sum \cite{Co08}, hence the two definitions coincide in this case.

 We aim to transform $\psi_0$.  All that is guaranteed for now for $\psi_0$, by Lemma\,\ref{lemma:contraction}, are exponential bounds in time, therefore we use \eqref{defPz} which is guaranteed to converge for $z$ small enough. We then prove in this section that, if the initial condition $f\in L_2$ then $\psi_0$ decays in time, while for the wave initial condition\-~(\ref{init}) $\psi_0$ approaches a periodic function.

The proposition below shows the form of a discrete-Laplace transformed integral operator with a kernel of the form in which we are interested here.

\begin{Proposition}\label{PerK}

Consider an operator of the form
\begin{equation}\label{opLform}
Lf(t)=\int_0^t ds\, f(s)K(s,t),
\end{equation} 
where $K(s,t)=0$ if $t<0$ or $s\not\in[0,t]$.
Then 
\begin{equation}\label{peropint}
 (\mathcal{P}_\sigma Lf)(\tau) =\int_{-\pi}^\pi dr \int_{0}^1 d\sigma_1\, (\mathcal{P}_{\sigma_1} f)(r)\, (\mathcal{P}_{-\sigma_1}\mathcal{P}_\sigma K)(r,\tau):=(\mathcal{P}_\sigma L) \mathcal{P}_{\sigma} f
 \end{equation} 
 
For complex $\sigma$ with $\Im\sigma>0$ the integrals $\int_0^1 d\sigma_1$ are replaced by $\frac1{2\pi i}\oint \frac{1}{z_1}dz_1$.

\end{Proposition}

\begin{proof}
An immediate calculation, using \eqref{defP}, then \eqref{invP} shows that
\begin{multline}\label{intop}
 (\mathcal{P}_\sigma Lf)(\tau) =\int_0^\tau ds\, f(s)K(s,\tau)\Theta(\tau)+\sum_{k=1}^\infty e^{i\sigma2k\pi}\int_0^{\tau+2k\pi}\, ds\, f(s)K(s,\tau+2k\pi)\\ 
 =\int_0^\tau ds\, f(s)K(s,\tau)\Theta(\tau)+\sum_{k=1}^\infty e^{i\sigma2k\pi}\left[ \int_0^\pi +\sum_{j=1}^{k-1}\int_{(2j-1)\pi}^{2(j+1)\pi}+ \int_{(2k-1)\pi}^{\tau+2k\pi}\right] ds\, f(s)K(s,\tau+2k\pi)\\
 =\int_0^\tau ds\, f(s)K(s,\tau)\Theta(\tau)+\sum_{k=1}^\infty e^{i\sigma2k\pi}\left[ \int_0^\pi ds\,f(s)K(s,\tau+2k\pi)\right.\\
 +\left. \sum_{j=1}^{k-1}\int_{-\pi}^{\pi}ds\,f(s+2j\pi)K(s+2j\pi,\tau+2k\pi)
+ \int_{-\pi}^{\tau}ds\,f(s+2k\pi)K(s+2k\pi,\tau+2k\pi)\right]  \\
= \int_0^\tau ds\,\int_{0}^{1}d\sigma_1 \mathcal{P}_{\sigma_1} f(s) K(s,\tau) + \sum_{k=1}^\infty e^{i\sigma2k\pi} \int_0^\pi ds\,\int_{0}^{1}d\sigma_1\, \mathcal{P}_{\sigma_1} f(s)K(s,\tau+2k\pi)\\
+ \sum_{k=1}^\infty \sum_{j=1}^{k-1}e^{i\sigma2k\pi} \int_{-\pi}^{\pi} ds\,\int_{0}^{1}d\sigma_1\, e^{-i\sigma_12j\pi}\mathcal{P}_{\sigma_1} f(s)K(s+2j\pi,\tau+2k\pi)\\
+ \sum_{k=1}^\infty e^{i\sigma2k\pi}  \int_{-\pi}^{\tau}ds\, \int_{0}^{1}d\sigma_1\, e^{-i\sigma_12k\pi}\mathcal{P}_{\sigma_1} f(s) K(s+2k\pi,\tau+2k\pi)  
\end{multline}
which equals the right hand side of \eqref{peropint}, since $K(s,t)=0$ if $s\not\in[0,t]$.
\end{proof}

\subsection{Analytic structure of solutions}\label{anstrsol}
We will apply the discrete-Laplace transform to integral kernels $K(s,t)$ which are multiples of $\tfrac1{\sqrt{t-s}}$, and this factor introduces singularities in $\sigma$. We start by treating $\tfrac1{\sqrt{t-s}}$ as a standalone term, as this clarifies the techniques needed, and then proceed with the actual operator.
 
 For $K(s,t)=\tfrac1{\sqrt{t-s}}\Theta(t)\chi_{[0,t)}(s)$ a direct calculation shows that the discrete-Laplace transformed kernel in \eqref{peropint} has the expression
\begin{multline}\label{sumsq}
 (\mathcal{P}_{-\sigma_1}\mathcal{P}_\sigma K)(r,\tau) =\Theta(\tau) \chi_{[0,\tau]}(r)\frac1{\sqrt{\tau-r}} + \chi_{[0,\pi]}(r) \sum_{k=1}^\infty \frac {e^{i\sigma 2k\pi}}{\sqrt{\tau+2k\pi-r}}  \\
+\sum_{k=1}^\infty \sum_{j=1}^{k-1} \frac {e^{i\sigma 2k\pi} e^{-i\sigma_1 2j\pi}}{\sqrt{\tau+2(k-j)\pi-r}} +  \chi_{[-\pi,\tau]}(r) \sum_{k=1}^\infty  \frac {e^{i\sigma 2k\pi} e^{-i\sigma_1 2k\pi}}{\sqrt{\tau-r}}
\end{multline}

Some of the series in \eqref{sumsq} must be interpreted in the sense of distributions. To see how, we truncate the series in $k$ to a term $N$ then take the limit $N\to\infty$. Take for example the third term in \eqref{sumsq}, the most involved. Changing the index of summation $j$ to $\ell=k-j$ we have

$$T_3=\sum_{k=1}^\infty \sum_{j=1}^{k-1} \frac {e^{i\sigma 2k\pi} e^{-i\sigma_1 2j\pi}}{\sqrt{\tau+2(k-j)\pi-r}} =\frac1{\sqrt{2\pi}} \sum_{k=1}^N\frac{u^k}{w^k}\sum_{\ell=1}^{k-1}\frac{w^\ell}{\sqrt{a+l}}$$
where
\begin{equation}\label{notuw}
u=e^{i\sigma 2\pi},\ \ w=e^{i \sigma_12\pi},\ \ a=\frac{\tau-r}{2\pi}
\end{equation}
For $u,w\ne 1$ (meaning that  $\sigma,\sigma_1\ne 0$) we use the integral representation of the Lerch transcendent\-~\cite[(25.14)]{DLMF1.1.6}:
\begin{equation}\label{Phiint}
\Phi\left(z,s,a\right)=\frac{1}{\Gamma\left(s\right)}\int_{0}^{\infty}\frac{p^{s-1}e^{-ap}}{1-ze^{-p}}\mathrm{d}p,\ \ \ z\not\in[1,+\infty)
\end{equation}
and the identity
\begin{equation}
\label{sumPhi}
\sum_{\ell=1}^{k-1}\frac{z^\ell}{(a+\ell)^b}=z\Phi(z,b,a+1)-z^k\Phi(z,b,a+k)
\end{equation}
We then have
\begin{multline}
T_3=\frac1{\sqrt{2\pi}} \sum_{k=1}^N \frac{u^k}{w^k}\left[w\Phi(w,\frac12,a+1)-w^k\Phi(w,\frac12,a+k)\right]\\
=\frac1{\sqrt{2\pi}}w \Phi(w,\frac12,a+1)\sum_{k=1}^N \frac{u^k}{w^k}-\frac1{\sqrt{2\pi}}\frac1{\sqrt{\pi }} \int_0^\infty \frac{dp}{\sqrt{p }}e^{-ap}\frac1{1-we^{-p}}\sum_{k=1}^Nu^ke^{-pk}
=T_{3,1}-T_{3,2}
\end{multline}
with
\begin{equation}
T_{3,1}
:= \frac1{\sqrt{2\pi}}w \Phi(w,\frac12,a+1)\sum_{k=1}^N \frac{u^k}{w^k}
,\quad
T_{3,2}:=\frac1{\sqrt{2\pi}}\frac{u}{\sqrt{\pi }} \int_0^\infty \frac{dp}{\sqrt{p }}e^{-(a+1)p}\frac1{1-we^{-p}}\frac{1- \left(ue^{-p}\right)^N}{1-ue^{-p}}
\end{equation}

To determine $\lim_{N\to\infty}T_{3,1}$, we note that it appears in \eqref{intop} in an integral form, after multiplication by the periodic function $\mathcal{P}_{\sigma_1} f$, then integrated in $\sigma_1$. We have, using \eqref{notuw}, for $\sigma,\sigma_1\ne0$,

\begin{multline}
\lim_{N\to\infty}   \int_0^1d\sigma_1 T_{3,1}\mathcal{P}_{\sigma_1} f(\tau)\\
 =\lim_{N\to\infty} \frac1{\sqrt{2\pi}}\sum_{k=1}^N e^{2k\pi i \sigma} \int_0^1d\sigma_1e^{-2k\pi i \sigma_1}e^{2\pi i \sigma_1} \Phi\left(e^{2\pi i \sigma_1},\frac12,a+1\right)\mathcal{P}_{\sigma_1} f(r) \\
 =\frac1{\sqrt{2\pi}}e^{2\pi i \sigma} \Phi\left(e^{2\pi i \sigma},\frac12,a+1\right)\mathcal{P}_{\sigma} f(r)-\frac1{\sqrt{2\pi}} \int_0^1d\sigma_1e^{2\pi i \sigma_1} \Phi\left(e^{2\pi i \sigma_1},\frac12,a+1\right)\mathcal{P}_{\sigma_1} f(r)
\end{multline}
(the limit is a distribution).

Clearly, for $u,w\ne1$,
\begin{multline}\label{limT32}
\lim_{N\to\infty} T_{3,2}=\frac1{\sqrt{2\pi}}\frac{u}{\sqrt{\pi }} \int_0^\infty \frac{dp}{\sqrt{p }}e^{-(a+1)p}\frac1{1-we^{-p}}\frac{1}{1-ue^{-p}}  \\
\end{multline}

The other terms in \eqref{sumsq} are similar and simpler.

\

We now show that $\sigma=0$ and $\sigma_1=0$ are indeed singularities, namely square root branch points. For this we define the operator for $\sigma$ in the upper complex plane, and take the limit $\Im\sigma\to 0$.

Clearly \eqref{limT32} still holds for $u,w$ complex with $|u|<1,|w|<1$.

We now  deform the path of integration: $\int_0^\infty=\tfrac12\int_\mathcal{C}$ where $\mathcal{C}$ is a Hankel contour around $[0,+\infty)$, namely a contour starting at $\infty-0i$, going around 0, and ending at $\infty+0i$. We further deform the path of integration to $\mathcal{C}_1$ so that the poles at $p=\ln w$ and $p=\ln u$ are now inside $\mathcal{C}_1$, in the process collecting the residues:
$$\lim_{N\to\infty}T_{3,2}=\frac u{2\sqrt{2}\pi} \int_{\mathcal{C}_1} \frac{dp}{\sqrt{p }}e^{-(a+1)p}\frac1{(1-we^{-p})(1-ue^{-p})} - \frac u{2\sqrt{2}\pi} 2\pi i \frac{u}{u-w}\left( \frac{u^{-a-1}}{\sqrt{\ln u}} - \frac{w^{-a-1}}{\sqrt{\ln w}}  \right)$$
Letting $\Im\sigma,\sigma_1\to 0$, the first integral above is an analytic function, while the sum of residues equals
$$\frac i{\sqrt{2}}\frac{e^{i\sigma 4\pi}}{e^{i\sigma 2\pi} -e^{i\sigma_1 2\pi} } \left( \frac{e^{-i(a+1)\sigma 2\pi}}{\sqrt{i\sigma 2\pi}} -  \frac{e^{-i(a+1)\sigma_1 2\pi}}{\sqrt{i\sigma_1 2\pi}} \right)$$
and is analytic (including when $\sigma=\sigma_1$) except for $\sigma=0$ and $\sigma_1=0$, where there are square root branch points.

The term $T_{3,1}$ is similar: we deform the path of integration of $\Phi$, $[0,+\infty)$, to $\mathcal{C}$, which is further deformed to $\mathcal{C}_1$ so that the pole at $p=\ln w=2\pi i\sigma_1$ is inside $\mathcal{C}_1$, in the process collecting the residue:
\begin{multline}
\Phi(w,\tfrac12,a+1)=\frac 1{\sqrt{\pi}}\int_0^\infty \frac{dp}{\sqrt{p}}  \frac{e^{-(a+1)p}}{1-we^{-p}} =\frac 1{2\sqrt{\pi}}\int_\mathcal{C} \frac{dp}{\sqrt{p}}  \frac{e^{-(a+1)p}}{1-we^{-p}}\\
=\frac 1{2\sqrt{\pi}}\int_{\mathcal{C}_1} \frac{dp}{\sqrt{p}}  \frac{e^{-(a+1)p}}{1-we^{-p}}+\frac 1{2\sqrt{\pi}}\, 2\pi i\, \frac{e^{-(a+1)i 2\pi\sigma_1 }}{\sqrt{i 2\pi\sigma_1}}:=\Phi_1(\sigma_1)
\end{multline}
and in this form we can let $\Im\sigma_1\to 0$. Taking the limit $N\to\infty$ as before, we obtain the limit as a distribution, which now, due to the residue, contains square root branch points.
We thus see that $\Phi_1(\sigma_1)$ has the form
\begin{equation}\label{anstr}
A_1(\sigma_1)+\frac{1}{\sqrt{\sigma_1}}A_2(\sigma_1),\ \ \ \ \text{with\ }A_{1,2}\ \text{analytic}
\end{equation}

\subsection{Solving the discrete-Laplace transformed equation \eqref{eqcontr}}\label{sovpereq}

We first note that the series for 
$\mathcal{P}_\sigma\psi_0(\tau)$ converges when $z:= {\rm e}^{2\pi i\sigma}$ has small enough absolute value, by Lemma\,\ref{lemma:contraction}. We show in Proposition \ref{prop:floquet} that the series converges for $|z|<1$ and that the only singularities are square root branch points at $\sigma=0$ and at $\sigma=\sigma_0$. Based on this, Lemma\,\ref{psi0longtime} provides the decay of $\psi(0,t)$ and finishes the proof of  Theorem\,\ref{Thlongtime} (ii).

The following Theorem establishes the analytic structure of $\mathcal{P}_\sigma\psi_0$.
We apply discrete-Laplace transform of Proposition\,\ref{PerK} to our integral equation \eqref{eqcontr} and obtain

\begin{Theorem}\label{sweatandblood}

Assume the initial condition $f(x):=\psi(x,0)$ is differentiable, with $f'\in L^2(\RR)$, that $f(0)=0$, $f$ has compact support and $\int_{-\infty}^0f(y)dy=0$, $\int_0^\infty f(y)dy=0$.

Let $\sigma_0$ be the fractional part of  $U+\tfrac{E^2}{4\omega^2}$.

Let $\psi_0$ be the unique solution of equation \eqref{eqcontr} given by Lemma\,\ref{lemma:contraction}.

For simplicity, we choose units such that $\omega=1$.

(a) The discrete-Laplace transform $\Psi_\sigma:=\mathcal{P}_\sigma\psi_0$ satisfies the equation
\begin{equation}
  \label{eq:inteq}
  \Psi_\sigma=\mathcal{L}_\sigma\Psi_\sigma+\mathcal{K}_\sigma\Psi_\sigma+\mathcal{P}_\sigma h
\end{equation}
where $\mathcal{L}_\sigma+\mathcal{K}_\sigma$ is the discrete-Laplace transform of the integral operator $L$ given in given by \eqref{L} and $h$ is given by \eqref{h}.

The operator $\mathcal{K}_\sigma$ is a sum of operators of the form

\begin{equation}
  \label{formK}
  \Psi_\sigma(\tau)\mapsto \int_{-\pi}^\pi ds\, \int_0^1d\sigma_1 \Psi_{\sigma_1}(s)H(s,\sigma_1,\tau,\sigma)
  \end{equation}
where $H$ is an analytic function for $\Im\sigma 
>0$ multiplying characteristic functions of intervals, and for real $\sigma$ having square root branch points at $\sigma=\sigma_0$, $\sigma_1=\sigma_0$, $\sigma_1=0$ and at $\sigma=0$ and analytic at all other points. 

The operator $\mathcal{L}_\sigma$ is a sum of operators of the form

\begin{equation}
  \label{formL}
  \Psi_\sigma(\tau)\mapsto \int_{-\pi}^\pi ds\,  \Psi_{\sigma}(s)F(\tau,s,\sigma)
  \end{equation}
where $F$ is an analytic function for $\Im\sigma 
>0$ multiplying characteristic functions of intervals, and for real $\sigma$ having square root branch points at $\sigma=\sigma_0$ and at $\sigma=0$  and analytic at all other points.

The operator $\mathcal{L}_\sigma$ is compact on $L^2([-\pi,\pi],d\tau)$ and $\mathcal{K}_\sigma$ is compact on $L^2([-\pi,\pi]\times [0,1),d\tau\, d\sigma)$.

 $\sqrt{\sigma-\sigma_0}\,\mathcal{K}_\sigma$ is analytic in $\sqrt{\sigma-\sigma_0}$ and 
$\mathcal{K}_\sigma$ has analytic continuation on the Riemann surface of the square root.

(b)  $\sqrt{\sigma-\sigma_0}\,\Psi_\sigma$ is analytic in $\sqrt{\sigma-\sigma_0}$.

\end{Theorem}

\

\begin{proof}

\

The outline of the proof is as follows. In \S\ref{1p2p1p} we calculate the discrete-Laplace transform of the integral operator and of the inhomogeneous term. The discrete-Laplace transformed operator, $\mathcal{L}_\sigma+\mathcal{K}_\sigma$, has a ``singular'' part,  $\mathcal{L}_\sigma$, which needs to be considered in a one-dimensional space, with $\sigma$ being a parameter. $\mathcal{K}_\sigma$ is a usual Fredholm operator in two dimensions. In \S\ref{calcPh} we calculate the discrete-Laplace transform of the inhomogeneous term $h$, finishing the proof of (a). To prove (b), we deduce the existence and analytic structure of $(I-\mathcal{L}_\sigma-\mathcal{K}_\sigma)^{-1}$ using the analytic Fredholm alternative as follows. First, in \S\ref{merosol}, we first apply the analytic Fredholm alternative to invert $I-\mathcal{L}_\sigma$ (operator in one variable). We then treat the resulting equation, \eqref{eq:inteq2}, by splitting it into a system (a regular part and a ''pole'' part) which we show has a meromorphic solution. In \S\ref{Proofofb} Lemma\,\ref{ifpoles} we show that any poles can only occur for $\sigma\in\RR$, and thus the series of $\mathcal{P}_\sigma \psi_0$ converges for $|z|<1$. Finally, in \S\ref{RAGE} we show that there are no poles even if $\sigma\in\RR$, proving (b).

\subsubsection{Calculation of $\mathcal{L}_\sigma, \mathcal{K}_\sigma$ and their analytic properties}\label{1p2p1p}

By Proposition\,\ref{PerK} the operator $\mathcal{L}_\sigma+\mathcal{K}_\sigma$ is the integral operator 
\begin{equation}
\label{fromK}
f(\tau)\mapsto \int_{-\pi}^\pi dr\int_0^1d\sigma_1f(\tau) \tilde{K}(\sigma_1,\sigma,r,\tau) 
\end{equation}
where $ \tilde{K}(\sigma_1,\sigma,r,\tau)= (\mathcal{P}_{-\sigma_1}\mathcal{P}_\sigma K)(r,\tau)$ is the discrete-Laplace transform of $K$, the kernel of the integral operator $L$. 

The kernel of $L$ is a sum of three terms. We detail below the calculations for one of them, namely the most delicate. The others are similar and simpler.

Consider the first term:
$$T_1(t):=\int_0^t\,ds\, \psi_0*s^{-1/2}\ G(s,t)=\int_0^t\, ds\, \psi_0(s)\int_s^t\,du\, \frac{G(u,t)}{\sqrt{u-s}}$$
where, applying $\mathcal{P}_\sigma$ in the variable $t$, see \eqref{defP}, we obtain
\begin{multline}
\mathcal{P}_\sigma T_1(\tau)=\int_0^\tau\, ds\, \psi_0(s)\int_s^\tau\,du\, \frac{G(u,\tau)}{\sqrt{u-s}}+\sum_{k=1}^\infty e^{i\sigma k 2\pi}\int_0^{\tau+2k\pi}ds\, \psi_0(s)\int_s^{\tau+2k\pi}\,du\, \frac{G(u,\tau+2k\pi)}{\sqrt{u-s}}\\
=\int_0^\tau\, ds\, \psi_0(s)\int_s^\tau\,du\, \frac{G(u,\tau)}{\sqrt{u-s}}+\sum_{k=1}^\infty e^{i\sigma k 2\pi}\int_0^{\pi}ds\, \psi_0(s)\int_s^{\tau+2k\pi}\,du\, \frac{G(u,\tau+2k\pi)}{\sqrt{u-s}}  \\
+\sum_{k=2}^\infty e^{i\sigma k 2\pi}\sum_{j=1}^{k-1} \int_{(2j-1)\pi}^{(2j+1)\pi}ds\, \psi_0(s)\int_s^{\tau+2k\pi}\,du\, \frac{G(u,\tau+2k\pi)}{\sqrt{u-s}} \\+\sum_{k=1}^\infty e^{i\sigma k 2\pi}\int_{(2k-1)\pi}^{\tau+2k\pi}ds\, \psi_0(s)\int_s^{\tau+2k\pi}\,du\, \frac{G(u,\tau+2k\pi)}{\sqrt{u-s}}:=T_{1,1}+T_{1,2}+T_{1,3}+T_{1,4}
\end{multline}

It suffices to establish the properties listed in a) for each of the terms above.
Let us look at the most involved of the terms $T_{1,j}$ above: changing the variable of integration we have
\begin{equation}
T_{1,3}=\sum_{k=2}^\infty e^{i\sigma k 2\pi}\sum_{j=1}^{k-1} \int_{-\pi}^{\pi}ds\, \psi_0(s+2j\pi)\int_{s+2j\pi}^{\tau+2k\pi}\,du\, \frac{G(u,\tau+2k\pi)}{\sqrt{u-s-2j\pi}}
\end{equation}
so
\begin{multline}
T_{1,3}= \sum_{k=2}^\infty e^{i\sigma k 2\pi}\sum_{j=1}^{k-1} \int_{-\pi}^{\pi}ds\, \psi_0(s+2j\pi)\int_{s+2j\pi}^{(2j+1)\pi}\,du\, \frac{G(u,\tau+2k\pi)}{\sqrt{u-s-2j\pi}}\\
+\sum_{k=3}^\infty e^{i\sigma k 2\pi}\sum_{j=1}^{k-1} \int_{-\pi}^{\pi}ds\, \psi_0(s+2j\pi) \sum_{m=j+1}^{k-1}\int_{(2m-1)\pi}^{(2m+1)\pi}\,du\, \frac{G(u,\tau+2k\pi)}{\sqrt{u-s-2j\pi}}\\
+\sum_{k=2}^\infty e^{i\sigma k 2\pi}\sum_{j=1}^{k-1} \int_{-\pi}^{\pi}ds\, \psi_0(s+2j\pi)\int_{(2k-1)\pi}^{\tau+2k\pi}\,du\, \frac{G(u,\tau+2k\pi)}{\sqrt{u-s-2j\pi}}\\
\label{T131}
\end{multline}
and so
\begin{multline}
T_{1,3}=\sum_{k=2}^\infty e^{i\sigma k 2\pi}\sum_{j=1}^{k-1} \int_{-\pi}^{\pi}ds\, \psi_0(s+2j\pi)\int_{s}^{\pi}\,dv\, \frac{G(v+2j\pi,\tau+2k\pi)}{\sqrt{v-s}}\\
+\sum_{k=3}^\infty e^{i\sigma k 2\pi}\sum_{j=1}^{k-1} \int_{-\pi}^{\pi}ds\, \psi_0(s+2j\pi) \sum_{m=j+1}^{k-1}\int_{-\pi}^{\pi}\,du\, \frac{G(v+2m\pi,\tau+2k\pi)}{\sqrt{u-s+2(m-j)\pi}}\\
\sum_{k=2}^\infty e^{i\sigma k 2\pi}\sum_{j=1}^{k-1} \int_{-\pi}^{\pi}ds\, \psi_0(s+2j\pi)\int_{-\pi}^{\tau}\,dv\, \frac{G(v+2k\pi,\tau+2k\pi)}{\sqrt{u-s+2(k-j)\pi}}\\
\end{multline}
which we split into the terms
\begin{equation}
T_{1,3}=: \int_{-\pi}^\pi ds\int_0^1d\sigma_1\mathcal{P}_{\sigma_1}(s)\left({\rm Sum}_1[G]+{\rm Sum}_2[G]+{\rm Sum}_3[G]\right)
\label{T13}
\end{equation}
where in the last step we used \eqref{invP}.

Now note that
\begin{equation}
\label{defG12}
G(s,t)=\frac{i\partial_sF_0(s,t)\, e^{iF_0(s,t)}}{\sqrt{t-s}}+\frac12 \frac{ e^{iF_0(s,t)}-1}{(t-s)^{3/2}}:=\frac{g_1(s,t)}{\sqrt{t-s}}+\frac{g_2(s,t)}{(t-s)^{3/2}}:=G_1+G_2
\end{equation}
and that
\begin{equation}
\label{perF_0}
F_0(s+2m\pi,\tau+2k\pi)=-(k-m)2\pi \tilde{U}+F_0(s,\tau),\ \ \ \ \ \  \ \partial_sF_0(s+2m\pi,\tau+2k\pi)=\partial_sF_0(s,\tau)
\end{equation}
where $ \tilde{U}=U+\frac{E^2}{4\omega^2}$, and $F_0,\ G$ are defined in \eqref{Fzero}, respectively \eqref{Gdef}.
Note that $\tilde U$ is the potential $U$ plus the ponderomotive energy \cite{Wo35}.

The sums in \eqref{T13} are split according to Sum$_j[G]$=Sum$_j[G_1]+$Sum$_j[G_2]$.

\

Calculation of Sum$_1[G_1]$ contains the main ingredients needed for the calculation of the others, so we start with this term, providing many details. From \eqref{defG12} and \eqref{perF_0} we see that $g_1(s+2j\pi,\tau+2k\pi)=e^{-i\tilde{U}(k-j)2\pi } g_1(s,\tau)$ thus we have
$${\rm Sum}_1[G_1]=\sum_{k=2}^\infty e^{i\sigma k 2\pi}\sum_{j=1}^{k-1} \int_{-\pi}^{\pi}ds\, \psi_0(s+2j\pi)\int_{s}^{\pi}\,dv\, \frac{G_1(v+2j\pi,\tau+2k\pi)}{\sqrt{v-s}}$$
\begin{multline}
=\int_{-\pi}^{\pi}ds\, \int_0^1d\sigma_1 \mathcal{P}_{\sigma_1}\psi_0(s)\,g_1(s,\tau) \sum_{k=2}^\infty e^{i\sigma k 2\pi}\sum_{j=1}^{k-1} e^{-i\sigma_1j2\pi}\int_{s}^{\pi}\,du\, \frac{e^{-i\tilde{U}(k-j)2\pi } }{\sqrt{u-s}} \frac1{\sqrt{\tau-u+(k-j)2\pi}}
\end{multline}
and the double sum above equals, after changing the index of summation from $j$ to $\ell=k-j$,
\begin{multline}
\sum_{k=2}^\infty e^{i\sigma k 2\pi}\sum_{\ell=1}^{k-1} e^{-i\sigma_1(k-\ell)2\pi}\int_{s}^{\pi}\,du\, \frac{e^{-i\tilde{U}\ell2\pi } }{\sqrt{u-s}} \frac1{\sqrt{\tau-u+\ell2\pi}}=\int_{s}^{\pi}\frac{du}{\sqrt{u-s}}  \sum_{k=2}^\infty e^{i(\sigma-\sigma_1) k 2\pi}\sum_{\ell=1}^{k-1}  \frac{e^{i(\sigma_1-\tilde{U})\ell2\pi } }{\sqrt{\tau-u+\ell2\pi}}\\
=\int_{s}^{\pi}\frac{du}{\sqrt{u-s}}  \sum_{k=2}^\infty e^{i(\sigma-\sigma_1) k 2\pi} \frac1{\sqrt{2\pi}}\left[ e^{i(\sigma_1-\tilde{U})2\pi } \Phi\left(e^{i(\sigma_1-\tilde{U})2\pi },\frac12,\frac{\tau-u}{2\pi}+1\right) \right.\\
\left.  - e^{i(\sigma_1-\tilde{U})k2\pi }\Phi\left(e^{i(\sigma_1-\tilde{U})2\pi },\frac12,\frac{\tau-u}{2\pi}+k\right)  \right]  \label{sumG1}
\end{multline}
where we used the  formula \eqref{sumPhi}.

The first sum in \eqref{sumG1} must be understood in the sense of distributions, and the second one is convergent. 
Indeed, for the first sum we have
\begin{multline}
\int_0^1d\sigma_1 \mathcal{P}_{\sigma_1}\psi_0(s)\,\int_{s}^{\pi}\,du\,\frac{g_1(s,\tau)}{\sqrt{u-s}}  \sum_{k=2}^\infty e^{i(\sigma-\sigma_1) k 2\pi} \frac1{\sqrt{2\pi}} e^{i(\sigma_1-\tilde{U})2\pi } \Phi\left(e^{i(\sigma_1-\tilde{U})2\pi },\frac12,\frac{\tau-u}{2\pi}+1\right)\\
=\mathcal{P}_{\sigma}\psi_0(s)\,\int_{s}^{\pi}\,du\,\frac{g_1(s,\tau)}{\sqrt{u-s}}  \frac1{\sqrt{2\pi}} e^{i(\sigma-\tilde{U})2\pi } \Phi\left(e^{i(\sigma-\tilde{U})2\pi },\frac12,\frac{\tau-u}{2\pi}+1\right):=\mathcal{P}_{\sigma}\psi_0(s)\,K_1(s,\tau)
\end{multline}
yielding a term in $\mathcal{L}_\sigma$, of the form \eqref{formL}. Since $K_1$ is continuous, the operator with this kernel is compact on $L^2([-\pi,\pi],ds)$.

To see that the second sum in \eqref{sumG1} is convergent, we use the integral representation \eqref{Phiint} of the Lerch $\Phi$ transcendendent; we have

\begin{multline}
\label{Ksigma1}
\int_{s}^{\pi}\frac{du}{\sqrt{u-s}}  \sum_{k=2}^\infty e^{i(\sigma-\sigma_1) k 2\pi} \frac1{\sqrt{2\pi}}
 e^{i(\sigma_1-\tilde{U})k2\pi }\Phi\left(e^{i(\sigma_1-\tilde{U})2\pi },\frac12,\frac{\tau-u}{2\pi}+k\right) \\
 =\int_{s}^{\pi}\frac{du}{\sqrt{u-s}}  \sum_{k=2}^\infty e^{i(\sigma-\sigma_1) k 2\pi} \frac1{\sqrt{2\pi}}
 e^{i(\sigma_1-\tilde{U})k2\pi } \int_0^\infty\frac{dp}{\sqrt{p}}\frac{e^{-p(\frac{\tau-u}{2\pi}+k)}}{1-e^{i(\sigma_1-\tilde{U})2\pi } e^{-p}}
\end{multline}
which is convergent for $\sigma_1-\tilde{U}\ne 0$, yielding a term in $\mathcal{K}_\sigma$, of the form \eqref{formK}. 

{\em Analytic structure.} For $\sigma_1-\tilde{U}= 0$ (meaning $\sigma_1=\sigma_0$) we proceed as in \S\ref{anstrsol}, only here the square root branch point will be at $\sigma_1=\sigma_0$ (instead of $\sigma_1=0$): we deform the path of integration and collect the residues.
The integral kernel \eqref{Ksigma1} has the form (analogue to \eqref{anstr})
\begin{equation}
\label{anstr2}
A_1(\sigma_1)+\frac{1}{\sqrt{\sigma_1-\sigma_0}}A_2(\sigma_1),\ \ \ \ \text{with\ }A_{1},A_2\ \text{analytic in }\sqrt{\sigma_1}
\end{equation}

The operator with the integral kernel \eqref{Ksigma1} is compact.

The calculation of Sum$_2[G]$ in \eqref{T13} is the most labor intensive, and we outline the main steps here (the details are as for the previous term). We change the order of summation: $\sum_{j=1}^{k-1}\sum_{m=j+1}^{k-1}=\sum_{m=1}^{k-2}\sum_{\ell=1}^{k-m-1}$
and using \eqref{perF_0} we obtain

\begin{multline}
{\rm Sum}_2[G_1]=\int_{-\pi}^\pi du\, g_1(u,\tau)\sum_{k=3}^\infty e^{i(\sigma-\sigma_1)2k\pi}\sum_{m=1}^{k-2}\frac{e^{i\sigma_12m\pi}}{\sqrt{u-s+2m\pi}} \sum_{\ell=1}^{k-m-1}\frac{e^{i(\sigma_1-\tilde{U})2\ell\pi}}{\sqrt{\tau-u+2\pi\ell}} \\
=\frac1{2\pi}\int_{-\pi}^\pi du\, g_1(u,\tau)\sum_{k=3}^\infty e^{i(\sigma-\sigma_1)2k\pi}\sum_{m=1}^{k-2}\frac{e^{i\sigma_12m\pi}}{\sqrt{\frac{u-s}{2\pi}+m}} \left[ e^{i(\sigma_1-\tilde{U})2\pi}\Phi\left(e^{i(\sigma_1-\tilde{U})2\pi},\frac12,\frac{\tau-u}{2\pi}+1\right)  \right.\\
\left. - e^{i(\sigma_1-\tilde{U})2(k-m)\pi}\Phi\left(e^{i(\sigma_1-\tilde{U})2\pi},\frac12,\frac{\tau-u}{2\pi}+k-m\right)  \right]:={\rm Term}_1+{\rm Term}_2
\end{multline}
and furthermore
\begin{multline}
{\rm Term}_1=\frac1{2\pi}\int_{-\pi}^\pi du\, g_1(u,\tau)\sum_{k=3}^\infty e^{i(\sigma-\sigma_1)2k\pi}   \\
\times  \left[ e^{i\sigma_12\pi}\Phi\left(e^{i\sigma_12\pi},\frac12,\frac{u-s}{2\pi}+1\right) - e^{i\sigma_1(k-1)2\pi}\Phi\left(e^{i\sigma_12\pi},\frac12,\frac{u-s}{2\pi}+k-1\right)   \right] \\
\times
 e^{i(\sigma_1-\tilde{U})2\pi}\Phi\left(e^{i(\sigma_1-\tilde{U})2\pi},\frac12,\frac{\tau-u}{2\pi}+1\right)\\
 =\frac1{2\pi}\int_{-\pi}^\pi du\, g_1(u,\tau) \delta_{\sigma-\sigma_1}   e^{i\sigma_12\pi}\Phi\left(e^{i\sigma_12\pi},\frac12,\frac{u-s}{2\pi}+1\right)  \  e^{i(\sigma_1-\tilde{U})2\pi}\Phi\left(e^{i(\sigma_1-\tilde{U})2\pi},\frac12,\frac{\tau-u}{2\pi}+1\right)\\
 -\frac1{2\pi}\int_{-\pi}^\pi du\, g_1(u,\tau)\sum_{k=3}^\infty e^{i(\sigma-\sigma_1)2k\pi}   \\
\times  \left[ e^{i\sigma_12\pi}\Phi\left(e^{i\sigma_12\pi},\frac12,\frac{u-s}{2\pi}+1\right) - e^{i\sigma_1(k-1)2\pi}\Phi\left(e^{i\sigma_12\pi},\frac12,\frac{u-s}{2\pi}+k-1\right)   \right] \\
\times
 e^{i(\sigma_1-\tilde{U})2\pi} \int_0^\infty \frac{dp}{\sqrt{p}} \frac{ e^{-\left(\frac{u-s}{2\pi}+k-1\right)p}}{1- e^{i\sigma_12\pi}e^{-p}}\\
 = \delta_{\sigma-\sigma_1} \frac1{2\pi}\int_{-\pi}^\pi du\, g_1(u,\tau)  e^{i\sigma 2\pi}\Phi\left(e^{i\sigma2\pi},\frac12,\frac{u-s}{2\pi}+1\right)  \  e^{i(\sigma-\tilde{U})2\pi}\Phi\left(e^{i(\sigma-\tilde{U})2\pi},\frac12,\frac{\tau-u}{2\pi}+1\right)+\text{analytic}
\end{multline}
and the first term above produces a term of the form \eqref{formL}, while the second term has the form \eqref{formK}.

Similarly, for $G_2$ we obtain the following term of the form \eqref{formL}:
\begin{multline}
\int_{-\pi}^\pi dr [\mathcal{P}_{\sigma}\psi_0](s)\, \frac1{2\pi}\int_{-\pi}^\pi du\,e^{i\sigma2\pi} \Phi\left(e^{i \sigma2\pi},\frac12,\frac{u-s}{2\pi}+1\right)\\
\times \left[ e^{iF_0(u,\tau)}e^{i(\sigma-\tilde{U})2\pi} \Phi\left(e^{i(\sigma-\tilde{U})2\pi},\frac32,\frac{\tau-u}{2\pi}+1\right)- e^{i\sigma2\pi} \Phi\left(e^{i\sigma2\pi},\frac32,\frac{\tau-u}{2\pi}+1\right)     \right]
\end{multline}
and a regular part, of the form \eqref{formK}.

The other terms are evaluated similarly and are simpler.

\subsubsection{Calculation of \,$\mathcal{P}_\sigma h$.}\label{calcPh}

We note the following identities:
$$\frac{e^{-iB/(n+a)}}{(n+a)^{1/2}}=\int_0^\infty\, dq\, e^{-nq}F_1(q),\ \ \frac{e^{-iB/(n+a)}}{(n+a)^{3/2}}=\int_0^\infty\, dq\, e^{-nq}F_2(q)$$
with 
\begin{equation}
\label{F12}
F_1(q)=\frac{e^{-aq} \cosh(2\sqrt{-iBq})}{\sqrt{\pi q}},\ \ F_2(q)=\frac{\sqrt{i}e^{-aq} \sinh(2\sqrt{-iBq})}{\sqrt{\pi B}}
\end{equation}
We saw that the kernels of $\mathcal{L}_\sigma,\mathcal{K}_\sigma $ have integral expressions. So will  also $\mathcal{P}_\sigma h$, except \eqref{notuw} is replaced by
$$u=e^{2\pi i (\sigma-\tilde{U})},\ \ w=e^{2\pi i (\sigma_1-\tilde{U})},\ \ \ \tilde{U}=U+\frac{E^2}{4\omega^2}$$
and instead of $\Phi(z,a,b)$ we have a sum of analytic functions multiplying $\int_0^\infty F_{1,2}(p)e^{-ap}(1-z e^{-p})^{-1}dp$, with the functions $F_1$ and $F_2$ given by \eqref{F12}.

Indeed, let us {\em calculate for the discrete-Laplace transform of} $h_+(0,t)$: with $\omega=1$ and the notations $\tilde{U}=U+E^2/4,\ A=E^2/8, \tilde{A}=E$ we have
\begin{multline}
\label{eq31}
\sqrt{2\pi i}\mathcal{P}_\sigma h_+(\tau)=\frac 1{\sqrt{\tau}}e^{-i\tilde{U}\tau+iA\sin 2\tau}\int_0^\infty dyf(y)e^{i\frac{(y+\tilde{A}(1-\cos\tau))^2}{2\tau}}\\
+\sum_{k=1}^\infty e^{i\sigma 2k\pi} e^{-i\tilde{U}(\tau+2k\pi)+iA\sin 2\tau}\int_0^\infty dyf(y)e^{i\frac{(y+\tilde{A}(1-\cos\tau))^2}{4\pi(\frac{\tau}{2\pi}+k)}}\\
=e^{-i\tilde{U}\tau+iA\sin 2\tau}\left[ \frac 1{\sqrt{\tau}}\int_0^\infty dyf(y)e^{i\frac{(y+\tilde{A}(1-\cos\tau))^2}{2\tau}}
+\sum_{k=1}^\infty e^{i(\sigma-\tilde{U}) 2k\pi} \int_0^\infty dyf(y)\int_0^\infty dq e^{-kq}F_1(q) \right]\\
=e^{-i\tilde{U}\tau+iA\sin 2\tau}\left[ \frac 1{\sqrt{\tau}}\int_0^\infty dyf(y)e^{i\frac{(y+\tilde{A}(1-\cos\tau))^2}{2\tau}}
+  \int_0^\infty dyf(y)\int_0^\infty dq F_1(q)\frac{e^{-q+i(\sigma-\tilde{U}) 2\pi}}{1-e^{-q+i(\sigma-\tilde{U}) 2\pi}} \right]
\end{multline}

Note that under our assumptions on $f$, the term
$$\frac 1{\sqrt{\tau}}\int_0^\infty dyf(y)e^{i\frac{(y+c)^2}{2\tau}},\ \ \ c=\tilde{A}(1-\cos\tau)$$
in the last line of \eqref{eq31} is in $L^2([-\pi,\pi],d\tau)$. To see this we integrate by parts, then change the variable of integration:
$$\frac 1{\sqrt{\tau}}\int_0^\infty dyf(y)e^{i\frac{(y+c)^2}{2\tau}}= \frac 1{\sqrt{\tau}}\int_0^\infty dyf'(y)\int_y^\infty du \,e^{i\frac{(u+c)^2}{2\tau}}=\frac1{\sqrt{2}} \int_0^\infty dyf'(y)\int_{\frac{(y+c)^2}{2\tau}}^\infty du \,\frac{e^{is}}{\sqrt{s}} $$
which is in $L^2$ since $f'\in L^2(\RR)$ and the integral $\int_{\nu}^\infty du \,\frac{e^{is}}{\sqrt{s}} $ is uniformly bounded (easily seen after an integration by parts).

{\em The discrete-Laplace transform of the term} $h_3(t):=\int_0^t(h_-*s^{-1/2})G(s,t)$ in $h$ yields singularities of the type studied in \S\ref{anstrsol}. Indeed 
\begin{equation}
\label{h3}
h_3(t)=\int_0^tdu\, h_-(u)\int_u^tds\frac1{\sqrt{s-u}}\left[ \frac 12 \frac{e^{iF_0(s,t)}-1}{(t-s)^{3/2}}+i\frac{\partial_s F_0. e^{iF_0(s,t)}}{\sqrt{t-s}}\right]
\end{equation}
which has a singularity of the type $\frac1{\sqrt{t-s}}$ which is preserved upon discrete-Laplace transform due to the special form of $F_0$, as seen in \S\ref{1p2p1p}. Indeed, by Proposition\,\ref{PerK} it suffices to discrete-Laplace transform the integral kernel in \eqref{h3}, which leads to a sum of terms of the form 
$$\frac{e^{iF_0(r+2j\pi,\tau+2k\pi)}-1}{(\tau+2k\pi-r-2j\pi)^{3/2}} =\frac{e^{-\tilde{U}(\tau+2k\pi-r-2j\pi)} e^{iF_0(r,\tau)}-1}{(\tau+2k\pi-r-2j\pi)^{3/2}}$$
which again, has a square root singularity.

 In the same way as in \S\ref{anstrsol} and \S\ref{1p2p1p} it follows that $\sqrt{\sigma-\sigma_0}\,\mathcal{P}_\sigma h$ is analytic in  $\sqrt{\sigma-\sigma_0}$ and in $\sqrt{\sigma}$.

\subsubsection{Existence of meromorphic solutions}\label{merosol}
Existence of solutions of \eqref{eq:inteq} for large $\Im\sigma>0$ follows from the existence of $\psi_0$, proved in Lemma\,\ref{lemma:contraction} and Propostion\,\ref{PerK}.

We showed in \S\,\ref{1p2p1p} that the operator  $\sqrt{\sigma-\sigma_0}\mathcal{L}_\sigma$ is analytic in $\sqrt{\sigma-\sigma_0}, \sqrt{\sigma}$ for $\sigma\ne 0$ and it is compact on $L^2[-\pi,\pi]$. Denote $z=e^{2\pi i \sigma}$; then $\mathcal{L}_\sigma$, $\mathcal{K}_\sigma$ are analytic in $z$, except for $z=z_0=e^{2\pi i \sigma_0}$, where there is a square root branch point. For $z\ne z_0$, by analytic Fredholm alternative $I-\mathcal{L}_\sigma$ has an inverse merormorphic  in $z$ and in a punctured neighborhood of each of its poles, say $z_p$, it has the form
$$(I-\mathcal{L}_\sigma)^{-1}=\frac1{(z-z_p)^m}M(z)+B(z)$$
where M is a finite rank operator, depending polynomially on $z$, and $B$ is analytic at $z_p$. Then $(I-\mathcal{L}_\sigma)^{-1}=\frac1{(z-z_p)^m}PM(z)+B(z)$ where $P$ is the orthogonal projection on ${\rm Ran}(M)$.

Applying this in \eqref{eq:inteq} we obtain
\begin{equation}
  \label{eq:inteq2}
  f=\frac1{(z-z_p)^m}PM \mathcal{K}_\sigma f+B\mathcal{K}_\sigma f +\left[\frac1{(z-z_p)^m}PM+B\right] h_\sigma,\ \ \ \text{where }f=\Psi_\sigma,\ h_\sigma=\mathcal{P}_\sigma h
\end{equation}
Denote by $P_\perp$ the orthogonal projection on Ran$(PM)$ . Then $f=Pf+P_\perp f$. Applying $P_\perp$ to \eqref{eq:inteq2} we obtain
$$P_\perp f=P_\perp B\mathcal{K}_\sigma( Pf+P_\perp f) +P_\perp B h$$
Now, $ \mathcal{K}_\sigma$ is compact on $\mathcal{H}=L^2([-\pi,\pi]\times[0,1),ds\,d\sigma_1)$. Then $P_\perp B\mathcal{K}_\sigma P_\perp$ is compact on $P_\perp\mathcal{H}$ has it has a meromorphic inverse, and there is $P_\perp f:=u$:
\begin{equation}
\label{defu}
u:=P_\perp f=(I_\perp - P_\perp B\mathcal{K}_\sigma P_\perp)^{-1}(P_\perp B\mathcal{K}_\sigma Pf+P_\perp B h_\sigma):= APf+\tilde{h}
\end{equation}

Now applying $P$ to \eqref{eq:inteq2} we obtain
\begin{equation}
\label{eqPf}
P f=\frac1{(z-z_p)^m}PM \mathcal{K}_\sigma (Pf+u)+PB\mathcal{K}_\sigma (Pf+u) +\left[\frac1{(z-z_p)^m}PM+PB\right] h_\sigma
\end{equation}
where, introducing $u$ from \eqref{defu} we obtain a finite dimensional equation for $Pf$, with meromorphic coefficients, which we know it has solutions. Therefore the solution $Pf$ of \eqref{eqPf} exists, and it is meromorphic in $z$.
We established that $(I-\mathcal{L}_\sigma-\mathcal{K}_\sigma)^{-1}$ is meromorphic in a neighborhood of the closed unit disk except for two square root branch points at $0$ and $\sigma_0$ and, in a neighborhood of any of its finitely many poles $z_p\notin\{0,\sigma_0\}$, it has the form
\begin{equation}
\label{eqinv}
(I-\mathcal{L}_\sigma-\mathcal{K}_\sigma)^{-1}=\frac1{(z-z_p)^{m_1}}M_p+B_p
\end{equation}
with $M_p$ of finite rank, polynomial in $z$ and $B_p$ analytic. Using analyticity in $\sqrt{\sigma}$ ($\sqrt{\sigma-\sigma_0}$ resp.), if a pole coincides with one of these branch points, then $m_1$ is simply replaced by $m+1/2$ and $M_p$ becomes analytic in  $\sqrt{\sigma}$ ($\sqrt{\sigma-\sigma_0}$ resp.).

\subsubsection{Poles imply nontrivial solutions of the quasienergy equation.} \label{Proofofb}
Lemma\,\ref{ifpoles} shows that if poles exist in \eqref{eqinv}, then there is a solution of  the Schr\"odinger equation \eqref{schrodinger2} with a special asymptotic behavior \eqref{asympsipol} in $t$.

\begin{Lemma}\label{ifpoles}
Let $z=e^{2\pi i\sigma}$, with $\Im\sigma>0$ (so that $|z|<1$).

Assume that $(I-\mathcal{L}_\sigma-\mathcal{K}_\sigma)^{-1}$ has a pole at $\sigma=\sigma_p$, that is, $M_p\ne 0$ in \eqref{eqinv}.
Then, for a dense set of initial conditions the Schr\"odinger equation \eqref{schrodinger2} has a solution of the form
\begin{equation}
\label{asympsipol}
\psi(x,t)=  t^{m-1}   e^{-i t\sigma_p} a(x,t)\left[ 1+O(1/t)\right]+O(1/\sqrt{t})
\end{equation}
with $a(x,\cdot)$ is $2\pi$-periodic, $a(\cdot,t)\in L^2(\RR)$ and also $\sigma_p\in\RR$.
\end{Lemma}

{\em Proof.}

Substituting \eqref{eqinv} in \eqref{eq:inteq} we obtain
 \begin{equation}
\label{psisigma}
 \Psi_\sigma(\tau)=\frac1{(z-z_p)^{m_p}}M_p\,\mathcal{P}_\sigma h + B_p\mathcal{P}_\sigma h
 \end{equation}
with
$$(I-\mathcal{L}_{\sigma_p}-\mathcal{K}_{\sigma_p})\,\mathcal{P}_{\sigma_p} h =0 $$

Let us simply denote $m_p=m$, $M_p=M,\ B_p=B$. 

We construct a $\psi_0$ so that $\Psi_\sigma$ of \eqref{psisigma} is its discrete-Laplace transform using  \eqref{defPz}, \eqref{invPz}.

Denoting $ \Psi_\sigma(\tau)=  F(z,\tau)$ we have 
$$\mathcal{P}_z\psi_0(\tau)=F(z,\tau)=\sum_{k=0}^\infty z^k\psi_0(\tau+2k\pi)$$
and the series converges in a disk $|z|<\delta$ by Lemma 6.
By \eqref{psisigma} $F$ has the form $F(z,\tau)=\frac{P(z;\tau)}{(z-z_p)^m}+g(z,\tau)$ where $P$ is a polynomial in $z$ of degree at most $m-1$ and $g$ is analytic at $z_p$. Then
$$\psi_0(\tau+2k\pi)=\frac 1{2\pi i}\oint_C \frac {F(z,\tau)}{z^{k+1}}\, dz$$
where $C$ is a closed path  containing $0$ inside the disk of radius $\delta$ and the pole $z_p$ is outside $C$. To determine large $k$ behavior we deform $C$ past the pole $z_p$ and leaving the path hanging around cuts at the branch points. In the process  we collect the residue at the pole, and then using the analytic properties of the operator, we push to two Hankel contours around the branch points $\sigma=0$ and $\sigma=\sigma_0$ linked by arccircles of radius $1+\epsilon$. 

{\em The contributions of the Hankel contours to the large $k$ behavior is $O(k^{-1/2})$.} Indeed, near $\sigma_0$, by Theorem\,\ref{sweatandblood} we have
\begin{equation}
\label{stareq20}
\mathcal{P}_\sigma\psi_0=\frac{a_{-1}}{\sqrt{\sigma-\sigma_0}}+a_1\sqrt{\sigma-\sigma_0}+f_1(\sigma)
\end{equation}
where $f_1$ is differentiable in $\sigma$. 
Integration by parts shows that
\begin{equation}\label{stareq2}
\int_0^1d\sigma\, e^{-i\sigma 2 k \pi}f_1(\sigma)=O(1/k)
\end{equation}
Hence 
\begin{equation}\label{largek}
\int_0^1 d\sigma\, e^{-i\sigma 2 k \pi} \mathcal{P}_\sigma\psi_0=-\frac{e^{3i\pi/4}\sqrt{2}a_{-1}}{2\sqrt{k}}\, {\rm erf}\left(\sqrt{2k\pi}e^{i\pi/4}\right)+O(1/k)\sim  -\frac{e^{3i\pi/4}a_{-1}}{\sqrt{2}} k^{-1/2} +O(1/k)
\end{equation}
where $O(1/k)$ comes from \eqref{stareq2} and from $\int_0^1 d\sigma\, e^{-i\sigma 2 k \pi}a_1\sqrt{\sigma-\sigma_0}$.

The contribution from the square root branch point at $\sigma=0$ is similar.

 {\em The contribution of the residues at the poles}, each of which is, to leading order, 
  \begin{equation}
    \label{largek}
\frac1{2\pi i} \oint_{|z-z_p|<\epsilon}dz\, \frac{P(z;\tau)}{z^{k+1}(z-z_p)^m} \sim (-1)^{m-1}k^{m-1}z_p^{-k-m}P(z_p,\tau)(1+O(1/k))
\end{equation}

Consider initial conditions $\psi(x,0)$ so that $\mathcal{P}_zh$ does not belong to $\cup_{j\ne p}Ran(M_j)$ where $j$ indexes the finitely many possible poles, and so that $\mathcal{P}_\sigma h\not\in$Ker$(M_p)$. Since $M_j$ are finite rank, this is a dense set of initial conditions.. For such initial conditions the leading order behavior of $\psi_0(\tau+2k\pi)$ is, with the notation $z_p=e^{2\pi i \sigma_p}$,
\begin{equation}
  \label{eq:asc}
 \psi_0(\tau+2\pi k) \sim k^{m-1}  e^{2\pi i \sigma_p(-k-m)} b(\tau)
\end{equation}

For $t>0$ let $\tau\in[-\pi,\pi)$ be the unique number so that $t=\tau+2k\pi$ with $k$ a positive integer. Then for large $t$

\begin{equation}
\label{psi0det}
\psi_0(t)\sim  t^{m-1}   e^{-i t\sigma_p} b_1(\tau)\left[ 1+O(1/t)\right]+O(1/\sqrt{t})
\end{equation}

For the assumed initial conditions as discussed the discrete-Laplace transform of $h$ exists up to the unit circle, and since the asymptotic form \eqref{psi0det} is still valid for $\sigma_1$ real.

By Lemma\,\ref{lemma:contraction} $\psi_0$ is continuously differentiable therefore $ b_1(\tau)$ extends to $[-\pi,\pi]$ periodically, hence to $a(t)$, a periodic functions of period $2\pi$.

Then $\psi(x,t)$, obtained by introducing $\psi_0$ in formulas \eqref{fomulapsix0}, \eqref{psiminus}, \eqref{psiplus} has a similar asymptotic form \eqref{psi0det}. To see this we note that convolutions with $t^{-1/2}$ preserves the asymptotic behavior \eqref{psi0det}  (up to multiplicative constants) since, expanding each $a_j(t):=a(t)$ in Fourier series $a(t)=\sum_kc_ke^{ikt}$ which converges uniformly since $a(t)$ is continuously differentiable, we see that
 \begin{multline}
\int_0^tds\, s^{m-1}e^{-i\sigma_p s}a(s)\frac1{\sqrt{t-s}}=t^{m-1}e^{-i\sigma_p t}\int_0^tds\,(1-s/t)^{m-1} e^{i\sigma_p s}a(t-s)\frac1{\sqrt{s}}\\
\sim t^{m-1}e^{-i\sigma_p t}\int_0^tds\, e^{i\sigma_p s}a(t-s)\frac1{\sqrt{s}}
=t^{m-1/2}e^{-i\sigma_p t}\int_0^1du\, e^{i\sigma_p ut}a(t(1-u))\frac1{\sqrt{u}}\\
=t^{m-1/2}e^{-i\sigma_p t}\sum_kc_k e^{ikt} \int_0^1 du\, e^{i(\sigma_p -k)ut}\frac1{\sqrt{u}}
\label{calca}
\end{multline}
where each integral in \eqref{calca} is evaluated by deforming the path of integration of the steepest descent and each integral is of order $t^{-1/2}k^{-1/2}$ and thus obtain that \eqref{calca} has dominant behavior $t^{m-1}e^{-i\sigma_p t}a_3(t)$ with $a_3$ is $2\pi$-periodic and smoothly differentiable.
Differentiation also preserves this form (being obtained from integral formulas, the asymptotic is differentiable).
 Furthermore, note that we required initial conditions so that $h_\pm(\cdot,t)$ decay sufficiently fast at $\infty$, thus being smaller than behavior \eqref{psi0det} for $m\ge 1$. The other integrals in \eqref{fomulapsix0}, \eqref{psiminus}, \eqref{psiplus} are treated similarly (recall that in this section we assumed $\omega=1$). Then, from \eqref{psiminus}, \eqref{psiplus}, $\psi(x,t)$ behaves, for large $t$, as a polynomial multiplying $e^{-i\sigma_p t}$ and $a(t)$, a $2\pi$-periodic function.

To study the behavior at $\sigma=\sigma_0$ we denote $\zeta=\sqrt{\sigma-\sigma_0}$ and we repeat the argument above, ruling out poles at $\zeta\ne 0$, while if there is a pole at $\zeta= 0$, it will have the form $\zeta^{-m}=(\sigma-\sigma_0)^{-m/2}$, which is not a pole in $\sigma$.

{\em We now show that $a(\cdot,t)\in L^2(\RR)$.}
The proof mimics the arguments in \S \ref{PfLema4}, (iii). An algebraically simpler way to see why this is to combine those arguments with the Fourier representations \eqref{eqhatpsim} and \eqref{eqhatpsip} below.  $\psi_0(t)$ converges in a space of differentiable functions with H\"older $1/4$ derivative, and, from \eqref{fomulapsix0}, $\psi_{x,0}$ converges in a space of functions  with H\"older exponent $3/4$  in intervals of the form $[t,t+2\pi/\omega]$. The norm in the latter space $\|f\|_{\infty}+\|f'\|_{\infty}+\sup_{x,y}|x-y|^{-\gamma}|f'(x)-f'(y)|$, with $\gamma=1/4$ and exponential weights are place on the sup norm as in \eqref{Bnu} to ensure contractivity. The integral operator is smoothing in this space.

  The integral term in \eqref{hatpsimin} converges uniformly in a space of functions on $\RR^+$ with values in $\{g:\|\sqrt{x^2+1} g\|_{\infty}<\infty\}$, hence uniformly a space of functions on $\RR^+$ with values in $L^2(\RR)$, to a $\tilde{\psi}\in L^2$ periodic in $t$ which solves \eqref{eqhatpsim}, as it is easy to check. (Note that the boundary condition at $x=0$ does not ensure symmetry of the Laplacian, nor hence conservation of  the $L^2$ norm.)

\

{\em  It remains to show that $\sigma_p$ is real. }Denote $\psi(x,t)=e^{-i\sigma_p t}\phi(x,t)$. Since $\psi$ satisfies the Schr\"odinger equation $i \psi_t=H\psi$ then $\phi$ satisfies $\sigma_p\phi+i \phi_t=H\phi$ therefore $\sigma_p\in\RR$, since the operator is symmetric.
 
 This completes the proof of Lemma\,\ref{ifpoles}. $\Box$.
 
 \end{proof}
 
 \

 {\bf Consequence.} Since there are no poles for $|z|<1$ (and no other singularities, by the Analytic Fredholm Alternative), the series of $\mathcal{P}_\sigma\psi_0$ converges for $|z|<1$.

\subsubsection{Absence of solutions of the quasienergy equation }\label{RAGE}

We first show that the existence of such solutions implies existence of actual eigenfunctions of the quasienergy operator; this implication is very general.

\begin{Lemma}\label{Hamsol}
  Consider a general Schr\"odinger equation
  \begin{equation}
    \label{eq:eqschr}
    i\psi_t(x,t)=H(x,t)\psi(x,t);\ \ x\in \RR^n
  \end{equation}
   where $H(x,t+2\pi)=H(x,t)$ for all $t$. Assume \eqref{eq:eqschr} has a solution of the form 
  \begin{equation}
    \label{eq:poly-per}
    \psi(x,t)=P(t)e^{i\lambda t}\phi(x,t),\ \ \ \text{where }\phi(x,t+2\pi)=\phi(x,t)\left[1+O(t^{-1})\right]+O(t^{-1/2})\ \ (t\to\infty)
  \end{equation}
where $P$ is a polynomial, $\psi$ is nonzero, $\lambda\in \RR$.
Then $P(t)$ is constant.
\end{Lemma}
\begin{proof}
 This follows from the fact that the evolution is unitary and $\|\psi(x,t)\|=1$ for all $t$.  
\end{proof}

\begin{Proposition}\label{prop:floquet}
  There are no nonzero solutions of satisfying \eqref{asympsipol}
  any $\lambda\in\mathbb \RR$.
  
  As a consequence, there are no poles for $z=e^{2\pi i \sigma}$ with $|z|\le 1$.
\end{Proposition}

{\em Proof.}

Recall that in this section we normalized equation so that $\omega=1$.

Consider a solution satisfying \eqref{asympsipol}. By Lemma\,\ref{Hamsol} we have $m=1$, therefore, with $\sigma=\lambda$,
\begin{equation}
\label{asympsipol1}
\psi(x,t)=   e^{-i \lambda t} a(x,t)\left[ 1+O(1/t)\right]+O(1/\sqrt{t})
\end{equation}
Substituting $\psi(x,t)=e^{i\lambda t}\tilde{\phi}(x,t)$ in the Schr\"odinger equation \eqref{schrodinger2} we see that $\tilde{\phi}$ solves:
\begin{equation}\label{eq11}
  i\partial_t\tilde{\phi}(x,t)-\frac 12\left[ -\partial_x^2+\Theta(x)(U-Ex\cos( t))\right]\tilde{\phi}(x,t)=\lambda \tilde{\phi}(x,t)
\end{equation}
where $\tilde{\phi}(\cdot,t)$ is in $L^2$ for each $t$ and for each $x$
\begin{equation}
\label{asyphi}
\tilde{\phi}(x,t)=a(x,t)\left[1+O(t^{-1})\right]+O(t^{-1/2})\ \ (t\to\infty),\ \ a(x,t)=a(x,t+2\pi),\ \forall t>0
\end{equation}
 We have $\lim_{t\to\infty}\tilde{\phi}(0,t)=a(0,t)$, periodic. We now show that there are no solutions with such matching conditions at $x=0$.

{\bf Remark 1.} If $\psi_0(t):=\psi(0,t)$ is $2\pi$-periodic then\ $\psi(x,\cdot)$ is $2\pi$-periodic.

{\em Indeed, }let $z=e^{2\pi is}$ where  $\Im s\ge 0$ to be complex. For $x<0$ we have $i\psi_t+\psi_{xx}=0$ with boundary condition $\psi_0(t)$.
 Now we write  $\phi(x,t)=e^{2\pi i\sigma t}\psi(x,t)$ we get $-2\pi\sigma\phi+\phi_t+\phi_{xx}=0$ where now $\phi(0,t)$ is periodic. Since for each fixed $x$ $\phi(x,t)$ is continuous, we take the discrete Fourier transform, 
$\phi(t)=\sum_{j\in\ZZ} C_j(x) e^{2\pi i j t}$ we get $-2\pi\sigma C_j -j C_j+C_j''=0$. The solution is $A_j e^{\pm \sqrt{2\pi j+s}x}$ where the sign depends on the sign of $\Re\sqrt{2\pi j+s}$ where $A_j$ are the Fourier coefficients of $\phi(0,t)$.  We note that the Fourier series of $\phi$ converges pointwise since $\phi$ is differentiable.The series  $\sum_jA_j e^{\pm \sqrt{2\pi j+s}}$ converges absolutely and uniformly since $|C_j|\le |A_j|e^{\pm \sqrt{2\pi j+s}x}$ where the sign ensures the real part is positive  and because of the convergence of the Fourier deries of $\phi$, $A_j\to 0$ as $j\to\infty$. We note that for such solutions to exist, we need that the Fourier coefficients $A_j$ vanish if $j$ is below a certain value. We have shown that, {\em if} such solutions exist, they are  analytic in $t$ for any $x\ne 0$ and periodic in $t$. The proof for $x>0$ is similar. The boundary condition becomes  $\sum C_j e^{\pm\sqrt{2\pi j+s} E\cos t+2\pi i j t}=\sum A_je^{2\pi i j t}$. Since $f_0$ is differentiable and hence $\sum A_je^{2\pi i j t}$ converges for all $t$, for $t=1$ we get that    $C_j e^{\mp\sqrt{2\pi j+s}}\to 0$ once more ensuring the absolute and uniform convergence of the series $\sum C_j e^{\pm\sqrt{2\pi j+s} \xi+i j t}$ in the corresponding domain (6.50). 

{\bf Remark 2.}  A straightforward but more tedious way is to rely on \eqref{eqhatpsim}, \eqref {hatpsimin} and \eqref{Cminus}, starting with $\lambda$ in the upper half plane to obtain, for $x<0$, an $L^2$ solution $\psi_0$ such that in the large $t$, $e^{i\lambda t}$ is periodic.  Similarly, for $x>0$, one uses \eqref{eqhatpsip}, \eqref{hatpsiplu}, \eqref{uPhi}, and \eqref{Cplus}.  

\

 We can equivalently work in the magnetic gauge. Let
\begin{equation}
  \varphi_t(x)=e^{-ixA_t\Theta(x)}\tilde{\phi}(x,t)
  ,\quad
  A_t:=\int_0^t d\tau\ E\cos( \tau)={E}\sin( t)
\end{equation}
Then $ \varphi_t(x)$ satisfies
\begin{equation}\label{eval}
  i\partial_t\varphi_t(x)-\frac12\left[ i\partial_x-\Theta(x)A_t\right]^2\varphi_t(x)-\Theta(x)U\varphi_t(x)=\lambda\psi(x,t)
  .
\end{equation}
The matching condition $\psi_t(0-)=\psi_t(0+),\ \partial_x\psi_t(0-)=\partial_x\psi_t(0+)$ becomes
\begin{equation}\label{matcon}
\varphi_t(0-)=\varphi_t(0+),\ \ \ \ \ \ \partial_x\varphi_t(0-)=\partial_x\phi_t(0+)+iA_t\varphi_t(0)
\end{equation}

We solve the equation \eqref{eval} for $x<0$ and $x>0$.

\noindent {\bf Negative $x$.}

For $x<0$ equation \eqref{eval} becomes 
$$ - i\partial_t\varphi_t(x)-\frac12\partial_x^2\varphi_t(x)=\lambda \varphi_t(x)$$
which we solve with boundary condition $\varphi_t(0)=a(0,t)$.

Substituting $\varphi_t(x)=\sum_{k\in\ZZ}u_k(x){\rm e}^{ik t}$ we obtain that $u_k(x)={\rm e}^{\pm \sqrt{2(k-\lambda)}\, x}$.

Solutions that decay towards $-\infty$ must have $k\omega-\lambda>0$ and the plus sign must be chosen at the exponent. Therefore, for $x<0$, 
\begin{equation}\label{solxneg}
\varphi_t(x)=\sum_{k\in\ZZ,k>\lambda}\ \ \ C_k\, {\rm e}^{ \sqrt{2(k-\lambda)}\, x}\,{\rm e}^{ik t}
\end{equation}
for some constants $C_k$.

\bigskip

\noindent {\bf Positive $x$. } For $x>0$ the equation  \eqref{eval}  becomes
\begin{equation}\label{eqxpos}
- i\partial_t\varphi_t(x) -\frac12\left(\partial_x^2 +2iA_t\partial_x-A_t^2-U\right)\varphi_t(x)=\lambda \varphi_t(x)
\end{equation}

{\bf Gauge transformation on a half-line; and eliminating the magnetic field.}
Substituting
$$u(x,t)={\rm e}^{g(t)}G(x+q(t),t)$$
with
$$ q(t)={E}\cos(t),\ \ \ \ \ g(t)={\frac {i{E}^{2}\sin \left( 2\,t \right) }{8}}
,\ \ \ \ \ \xi=x+q(t)$$
equation \eqref{eqxpos} becomes
\begin{equation}
  \label{eq:eqG1}
  - i\partial_tG(\xi,t)-\frac12\partial_x^2 G(\xi,t)=\left(-\lambda-U-\frac{E^2}{4}\right)G(\xi,t)=:\tilde{\lambda}G(\xi,t)
\end{equation}
The new PDE is defined on the domain
\begin{equation}
  \label{eq:defdom}
\mathcal{D}=  \left\{(\xi,t):\,t\ge 0,\ \xi+ {E} \cos( t)\ge 0\right\}
\end{equation}
It is clear that, for each fixed $t$,  the change of variables is an isomorphism between $L^2((-{E} \cos( t),\infty)$ and $L^2(\RR^+)$. We are looking for periodic solutions of \eqref{eq:eqG1}. Such solutions have Fourier series, convergent in $\mathcal D$:
\begin{equation}
  \label{eq:fourier0}
  G(\xi,t)=\sum_{n\in\ZZ}c_n(\xi) e^{i n t}
\end{equation}
Substituting \eqref{eq:fourier0} in \eqref{eq:eqG1} we obtain that for any $n\in\ZZ$ there is a $D_n\in\CC$ such that
\begin{equation}
  \label{eq:fourier0}
 c_n(\xi)=D_ne^{-\xi \sqrt{2(n-\tilde\lambda)}};\ \ n> n_0:=\tilde{\lambda}\ \ \text{and} \ c_n=0 \ \text{otherwise}
\end{equation}
hence
\begin{equation}
  \label{eq:formG}
  G(\xi,t)=\sum_{n> n_0}D_n e^{-\xi \sqrt{2(n-\tilde\lambda)}} e^{i n t}
\end{equation}
Since $G$ is differentiable the series converges pointwise convergence in the interior of $\mathcal D$, which implies 
\begin{equation}\label{eq:condconv}
| D_n|<\mathrm{Const}\  e^{\xi \sqrt{2(n-\tilde\lambda)}}
\end{equation}
The best bound is obtained when $t=(2m+1) \pi$  ($m\in\ZZ$),  see \eqref{eq:defdom},
\begin{equation}\label{eq:condconv}
 | D_n|<\mathrm{Const}\  e^{-E \sqrt{2(n-\tilde\lambda)}}
\end{equation}
We note that this estimate implies that the series  \eqref{eq:formG} converges uniformly and absolutely to a locally analytic function in the interior of $\mathcal{D}$, and it also converges uniformly and absolutely, together with all derivatives to its boundary, except perhaps at the special points $(-E \omega^{-2},(2m+1)\pi),\ m\in\NN$.

Returning to the variables $(x,t)$ we obtain 

\begin{equation}\label{solxpos}
\varphi_t(x)=\sum_{n\in\ZZ,n>n_0}\ \ \ D_n\, f_n(t)\, {\rm e}^{-\kappa_n  x}\, {\rm e}^{in t}
 \end{equation}
 \begin{equation}\label{notation}
 f_n(t)= {\rm e}^{\frac{i E^2}{8}\sin(2t)-\kappa_n   {E}\cos( t)},\ \ \ \ n_0=(-\lambda -U-\tfrac{E^2}{4}),\ \ \kappa_n=\sqrt{2}\sqrt{n+ \lambda+U+\frac{E^2}{4}}
 \end{equation}
 and convergence and analyticity are inherited from the above, for all $x>0,t>0$, all the way to $x=0$ except for the points $(0,(2m+1)\pi)$. 

We now show,  by contradiction, that \eqref{eval} has no nonzero solutions, 

We impose the matching conditions \eqref{matcon} for $\varphi_t(x)$ given by \eqref{solxneg} for $x<0$ and by \eqref{solxpos} for $x>0$. We must have
\begin{equation}\label{cond1}
\sum_{k\in\ZZ,k>\lambda}\ \ \ C_k\, \,{\rm e}^{ikt}=\sum_{n\in\ZZ,n>n_0}\ \ \ D_n\, {\rm e}^{in t}\, f_n(t):=\Phi(t)
\end{equation}

This equation holds pointwise except for $t=(2m+1)\pi:=t_m$, which means that, except at these points we are dealing with locally analytic functions of $t$, and the series on the left also converge pointwise uniformly a.e. (more precisely, except at $t_m$). 
From \eqref{cond1}, since we have $C_k=0$ for $k< \lambda$, then
$$\int_0^{2\pi}{\rm e}^{-ikt}\ \Phi(t)\, dt=0\ \ \ \ \text{for all }k< \lambda  $$
which we now show it is not possible unless all the $D_n=0$ in the sum.

Indeed, since the Fourier coefficients of $\Phi(t)$ vanish for $k<\lambda$, then $\Phi$ extends as a meromorphic function inside the disk bounded by $\mathbb{T}$. Denoting $z={\rm e}^{i t}$, the function $\Phi$ is presented as a convergent transseries (see e.g. \cite{Co08}) at $z=0$: 

\begin{equation}\label{tranP}
\Phi={\rm e}^{\frac{\epsilon^2}{2}(z^2-\frac{1}{z^2})} \sum_{n\in\ZZ,n>n_0}\ \ \ D_n\, z^n {\rm e}^{-\kappa_n   {4\epsilon}(\frac z2+\frac1{2z})} ={\rm e}^{-\frac{\epsilon^2}{2}\, {\frac{1}{z^2}} }\sum_{n>n_0}\  {\rm e}^{-\kappa_n   {2\epsilon}\, \frac1{2z}}\, g_n(z)
\end{equation}
with $g_n$ meromorphic and $\kappa_n$ strictly increasing in $n$. When transseries representations  exist, they are unique. Since $\Phi$ is meromorphic at $z=0$, the transseries representation \eqref{tranP} is possible only if all  $g_n,n\in\NN$ are zero, therefore $\Phi\equiv 0$. 

In conclusion $a(x,t)=0$, hence no poles can exist.

$\Box$

{\bf Note.} In the process we showed that we could work with the dominant term in \eqref{asyphi}, asymptotically, as $t\to\infty$.

\subsection{End of proof of Theorem \ref{Thlongtime}}\label{TheEnd}
Assume $x$ is in a compact set and $\psi(x,0)\in L^2$. The fact that decay of $\psi(x,t)$ is at least as fast as $t^{-1/2}$ follows from the explicit formula for $\psi(x,t)$ in terms of $\psi(x,0)$ and the following:

\begin{Lemma}\label{psi0longtime}

Assume $\psi(x,0)\in L^2$.

(i) We have $\psi(0,t)=O(t^{-1/2})$ as $t\to\infty$.

(ii) For $x$ in a compact set $\psi(x,t)=O(t^{-1/2})$.
\end{Lemma}

\begin{proof}
(i) The absence of poles proven in Proposition \ref{prop:floquet} shows that the main large $k$ asymptotic behavior of $\psi(0,\tau+2 k \pi \omega)$ comes from the Hankel contours around the branch points, namely \eqref{largek} resulting in $O(t^{-1/2})$ decay in $x$, uniformly on compact sets. (Uniformity follows immediately from \eqref{psiminus} and \eqref{psiplus}.)

(ii) The same arguments as in \S\ref{Proofofb} show that $\psi(x,t)=O(t^{-1/2})$.
\end{proof}

Since $\psi(x,t)=O(t^{-1/2})$ uniformly on compact sets in $\RR$, formula \eqref{intpsisq} follows.

Theorem\,\ref{Thlongtime} is proved.

 \begin{Remark}\label{rempoleswave} Note that starting with distributional (plane wave) initial condition  \eqref{init}, poles appear in $h_\pm$ as seen by a straightforward calculation and decay to an eventually periodic solution obtained by physical arguments by Faisal \cite{FKS05}).
 \end{Remark}

\subsection{Computation of $h_\pm$}
Let us start by computing $h_\pm$.

\begin{Proposition}\label{NP2}
For the initial condition $f(x)=\psi(0,x)$ in \eqref{init} we have
\begin{equation}\label{h0texp}
h_-(0,t)=\frac{e^{-i\frac{k^2t}2}}2
\left[
  \mathrm{erfc}\left({\textstyle-\sqrt{\frac t{2i}}k}\right)
  +R_0\, \mathrm{erfc}\left({\textstyle\sqrt{\frac t{2i}}k}\right)
\right]
\end{equation}
and
\begin{equation}\label{h0texpplus}
  h_+(0,t)=
  \frac{T_0}2
  e^{
    \frac{E}{\omega^2}(1-\cos(\omega t))\sqrt{2U-k^2}
    -\frac i2(k^2+\frac{E^2}{2\omega^2})t
    +i\frac{E^2}{8\omega^3}\sin(2\omega t)
  }
  \mathrm{erfc}\left({\textstyle \sqrt{\frac{it}2}\sqrt{2U-k^2}+\frac{E}{\omega^2\sqrt{2it}}(1-\cos(\omega t))}\right)
  .
  \end{equation}
\end{Proposition}

\begin{proof}
We first compute $h_-$: by\-~\eqref{hminus},
\begin{equation}
  h_-(t)=\sqrt{\frac1{2\pi i t}}\int_{-\infty}^0dy\ (e^{iky}+R_0e^{-iky})e^{i\frac{y^2}{2t}}
  .
\end{equation}
We have
\begin{equation}
  \begin{array}{>\displaystyle l}
    \int_{-\infty}^0dy\ e^{iky}e^{i\frac{y^2}{2t}}
    =
    e^{-i\frac{k^2t}{2}}\int_{-\infty}^0dy\ e^{i\frac1{2t}(y+tk)^2}
    =
    e^{-i\frac{k^2t}{2}}\int_{-tk}^\infty dy\ e^{i\frac1{2t}y^2}
    \\[0.5cm]\hfill
    =
    \sqrt{2it}e^{-i\frac{k^2t}{2}}
    \int_{-\sqrt{\frac{t}{2i}}k}^\infty dy\ e^{-y^2}
    =
    \sqrt{\frac{i\pi t}{2}}e^{-i\frac{k^2t}{2}}
    \mathrm{erfc}({\textstyle-\sqrt{\frac t{2i}}k})
    .
  \end{array}
\end{equation}
Therefore,
\begin{equation}
  h_-(t)=\frac{e^{-i\frac{k^2}{2}t}}2\left(
    \mathrm{erfc}({\textstyle-e^{-\frac{i\pi}4}\sqrt t\frac k{\sqrt{2}}})
    +R_0
    \mathrm{erfc}({\textstyle e^{-\frac{i\pi}4}\sqrt t\frac k{\sqrt{2}}})
  \right)
  .
  \label{h-}
\end{equation}
\bigskip

\indent
We now turn to $h_+$: by\-~\eqref{hplus}, if $\mathfrak q_k:=\sqrt{2U-k^2}$, then
\begin{equation}
  h_+(t):=T_0\sqrt{\frac1{2\pi i t}}e^{-i(U+\frac{E^2}{4\omega^2})t+i\frac{E^2}{8\omega^3}\sin(2\omega t)}\int_0^\infty dy\ e^{-\mathfrak q_ky}e^{i\frac{(y+\frac{E}{\omega^2}(1-\cos(\omega t)))^2}{2t}}
  .
\end{equation}
We have
\begin{equation}
  \int_0^\infty dy\ e^{-\mathfrak q_ky}e^{i\frac{(y+\frac{E}{\omega^2}(1-\cos(\omega t)))^2}{2t}}
  =
  e^{i\frac{E^2}{2\omega^4t}(1-\cos(\omega t))^2}
  \int_0^\infty dy\ e^{-\mathfrak q_ky}e^{i\frac1{2t}y^2+2i\frac{E y}{2\omega^2 t}(1-\cos(\omega t))}
\end{equation}
and
\begin{equation}
  \int_0^\infty dy\ e^{-\mathfrak q_ky}e^{i\frac1{2t}y^2+2i\frac{E y}{2\omega^2 t}(1-\cos(\omega t))}
  =
  e^{-i\frac 1{2t}(-it\mathfrak q_k-\frac{E}{\omega^2}(1-\cos(\omega t)))^2}
  \int_0^\infty dy\ e^{i\frac1{2t}(y+it\mathfrak q_k+\frac{E}{\omega^2}(1-\cos(\omega t)))^2}
\end{equation}
so
\begin{equation}
  \begin{array}{>\displaystyle l}
    \int_0^\infty dy\ e^{-\mathfrak q_ky}e^{i\frac{(y+\frac{E}{\omega^2}(1-\cos(\omega t)))^2}{2t}}
    =
    e^{\frac{it\mathfrak q_k^2}{2}+\frac{\mathfrak q_kE}{\omega^2}(1-\cos(\omega t))}
    \int_{it\mathfrak q_k+\frac{E}{\omega^2}(1-\cos(\omega t))}^\infty dy\ e^{i\frac1{2t}y^2}
    \\[0.5cm]\hfill
    =
    \sqrt{2it}e^{\frac{it\mathfrak q_k^2}{2}+\frac{\mathfrak q_kE}{\omega^2}(1-\cos(\omega t))}
    \int_{\sqrt{\frac{it}{2}}\mathfrak q_k+\frac{E}{\omega^2\sqrt{2it}}(1-\cos(\omega t))}^\infty dy\ e^{-y^2}
    \\[0.5cm]\hfill
    =
    \sqrt{\frac{i\pi t}{2}}e^{\frac{it\mathfrak q_k^2}{2}+\frac{\mathfrak q_kE}{\omega^2}(1-\cos(\omega t))}
    \mathrm{erfc}({\textstyle \sqrt{\frac{it}{2}}\mathfrak q_k+\frac{E}{\omega^2\sqrt{2it}}(1-\cos(\omega t))})
    .
  \end{array}
\end{equation}
Therefore,
\begin{equation}
  h_+(t)=\frac{T_0}2e^{\frac{E}{\omega^2}(1-\cos(\omega t))\mathfrak q_k-i\frac{k^2+\frac{E^2}{2\omega^2}}{2}t+i\frac{E^2}{8\omega^3}\sin(2\omega t)}\mathrm{erfc}(e^{-\frac{i\pi}4}({\textstyle i\sqrt t\frac{\mathfrak q_k}{\sqrt{2}}+E\frac{1-\cos(\omega t)}{\omega^2\sqrt t\sqrt{2}}}))
  .
  \label{h+}
\end{equation}
\bigskip
\end{proof}

\subsection{Poles of $\mathcal P_\sigma h_\pm$}
\indent
We now compute the poles of $\mathcal P_\sigma h_\pm$.

\begin{Proposition}\label{singPh}
  $\mathcal P_\sigma h_\pm$ has poles at $\sigma_p=\frac{k^2}2$ and $\sqrt\sigma\mathcal P_\sigma h_\pm$ is analytic in $\sqrt\sigma$.
\end{Proposition}

\begin{proof}
We start with $\mathcal P_\sigma h_-$.
Using the asymptotic expansion of the error function\-~\cite[(7.12.1)]{DLMF1.1.6},
\begin{equation}
  h_-(0,t)=
  \frac{e^{-i\frac{k^2t}2}}{\sqrt\pi}
  +
  \mathfrak R_-(t)
\end{equation}
with
\begin{equation}
  \mathfrak R_-(t):=
  -\frac{\sqrt i}{k\sqrt{2\pi t}}(1-R_0)
  +
  O(t^{-\frac32})
  .
\end{equation}
Proceeding as in section\-~\ref{anstrsol}, we find that $\sqrt\sigma\mathcal P_\sigma\mathcal R_-$ is analytic in $\sqrt\sigma$.
Furthermore,
\begin{equation}
  \mathcal P_\sigma\mathfrak h_-
  =
  \sum_{n=0}^\infty
  e^{2i\pi\sigma n}
  \frac{e^{-i\frac{k^2(t+2\pi n)}2}}{\sqrt\pi}
  =
  e^{-i\frac{k^2}2t}\frac1{\sqrt\pi(1-e^{2i\pi(\frac{k^2}2-\sigma)})}
\end{equation}
which has a pole at $\sigma=\frac{k^2}2$.
\bigskip

\indent
We now turn to $\mathcal P_\sigma h_+$.
By\-~\cite[(7.12.1)]{DLMF1.1.6},
\begin{equation}
  h_+(0,t)=
  \frac{T_0}{2\sqrt{i\pi t(U-\frac{k^2}2)}}
  e^{-it(U+\frac{E^2}{4\omega^2})+i\frac{E^2}{8\omega^3}\sin(2\omega t)}
  +O(t^{-\frac32})
  .
\end{equation}
Again, proceeding as in section\-~\ref{anstrsol}, we find that $\sqrt\sigma\mathcal P_\sigma\mathcal R_-$ is analytic in $\sqrt\sigma$.
\end{proof}

\subsection{End of proof of Theorem \ref{theo:2}}\label{lasttheo}
As we explained at the beginning of \S\ref{PfLongt}, we only need to take into account that for the distributional plane wave initial condition,  $\mathcal P_\sigma h_\pm$ has a pole at $\sigma_p=\frac{k^2}2$. Otherwise, $\sqrt\sigma\mathcal P_\sigma h_\pm$ is analytic in $\sqrt\sigma$ everywhere else (see Proposition \ref{singPh}).
Proceeding as in the proof of Lemma\-~\ref{ifpoles}, we find that the solution of the Schr\"odinger equation\-~\eqref{schrodinger2} is of the form
\begin{equation}
  \psi(x,t)=e^{-it\frac{k^2}2}\phi(x,t)(1+O(t^{-\frac12}))
\end{equation}
where $\phi(x,\cdot)$ is $2\pi/\omega$-periodic, which proves the theorem.
 \qed

\appendix

\section{Laplace transform versus discrete-Laplace transform}\label{appen}

In a way similar to the classical Poisson summation formula approach, working in distributions, taking a Laplace transform, which we denote by $\mathcal L$, followed by a discrete Fourier transform is related to a discrete-Laplace transform in the original variable, as seen below.
\begin{equation}
\label{genPos}
\frac1{2\pi}\sum_{n\in\ZZ}(\mathcal{L}\psi)(-i\sigma-in\omega){\rm e}^{-i(\sigma+n\omega)r}
=\frac1\omega\left[ \psi(r)\Theta(r)+\sum_{k=1}^\infty {\rm e}^{i\sigma\frac{2k\pi}\omega}\psi\left(r+\frac{2k\pi}\omega\right)\right]:=(\mathcal{P}_\sigma\psi)(r)
\end{equation}
where $-\tfrac\pi\omega\le r<\tfrac\pi\omega$ and $\sigma\in[0,\omega)$.

To deduce this formula,
we calculate
\begin{multline}
\frac1{2\pi}\sum_{n\in\ZZ}(\mathcal{L}\psi)(-i\sigma-in\omega){\rm e}^{-in\omega r}=\frac1{2\pi}\sum_{n\in\ZZ}\int_0^\infty e^{(i\sigma+in\omega)t}\psi(t){\rm e}^{-in\omega r}\, dt\\
=\frac1{2\pi}\sum_{n\in\ZZ}\left[\int_0^{\pi/\omega} e^{(i\sigma+in\omega)t}\psi(t){\rm e}^{-in\omega r}\, dt +\sum_{k=1}^\infty \int_{(2k-1)\pi/\omega}^{(2k+1)\pi/\omega} e^{(i\sigma+in\omega)t}\psi(t){\rm e}^{-in\omega r}\, dt    \right]\\
=\frac1{2\pi}\sum_{n\in\ZZ}\int_{-{\pi/\omega}}^{\pi/\omega}e^{i\sigma t}\psi(t) \Theta(t)\,e^{in\omega(t-r)}\, dt +\frac1{2\pi}\sum_{n\in\ZZ}\sum_{k=1}^\infty \int_{-\pi/\omega}^{\pi/\omega}  e^{i\sigma(s+\frac{2k\pi}\omega)} \psi(s+\frac{2k\pi}\omega)e^{in\omega(s-r+\frac{2k\pi}\omega)}  \\
\end{multline}
where we let $t=\tfrac{2k\pi}\omega+s$. Using the fact that $\tfrac1{2\pi}\sum_ne^{in\omega(s-r)}=\frac1\omega\delta_{s-r}$ formula \eqref{genPos} follows.

\section{Figures}\label{appendix:figs}

As already mentioned in the introduction Eq.  \eqref{schrodinger2} is the underlying basic model used for the interpretation of experiments of electron emission from a metal surface irradiated by lasers of different frequencies \cite{FN28,FKS05,HKK06,KSH11,YGR11,Ba06,KSe12b,PA12,YHe13,CPe14,ZL16,Fo16,Je17,KLe18,LZZ21}.
  This is so despite the fact that the system described by \eqref{schrodinger2} is very idealized, both in the description of the metal and in the use of a classical electric field.
 The literature therefore contains many approximate qualitative solutions of \eqref{schrodinger2} or some modification of it.
 Our analysis which proves the existence of physical solutions to \eqref{schrodinger2} does not give a visualization of the form of such solutions.
 To do that requires carefully controlled numerical solutions.
  Figure \ref{Fig2} shows the complex behavior of the current at early times for large fields.
 Figure \ref{Fig3} shows the steep rise of the current as the frequency of the applied field crosses the field dependent critical frequency, which is the energy that is necessary for an electron to absorb in order to be extracted from the metal: it is the real solution to the cubic equation $\omega_c=U-\frac{k^2}2+\frac{E^2}{4\omega^2_c}$ (the term $\frac{E^2}{4\omega_c^2}$ comes from the ``Zitterbewegung'' \cite{Wo35}).
 For small $E$, this reproduces the usual physical picture of the photoelectric effect.

The figures are obtained by solving the integral equation numerically for $\psi(x,t)$ with controlled approximations \cite{CCe20}.

\begin{figure}
  
  \hfil\includegraphics[width=8cm]{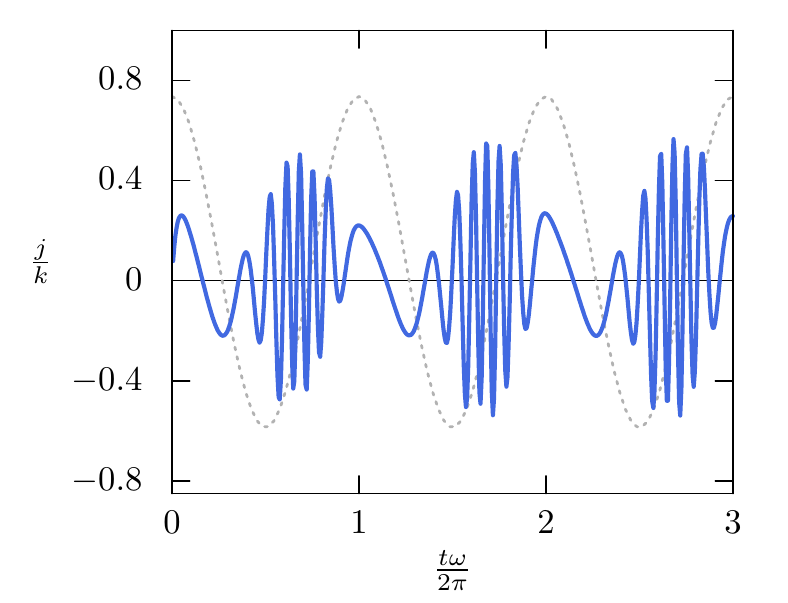}
  \caption{
    The normalized current $\frac jk$ at the interface (in atomic units, so $\frac jk$ is dimensionless) as a function of $\frac{t\omega}{2\pi}$ for $\omega=1.55\ \mathrm{eV}$ and for the electric field: $E=25\ \mathrm{V}\cdot\mathrm{nm}^{-1}$.
    The dotted line is the graph of $\cos(\omega t)$ (not to scale).
  }
  \label{Fig2}
\end{figure}

\begin{figure}
  \hfil\includegraphics[width=8cm]{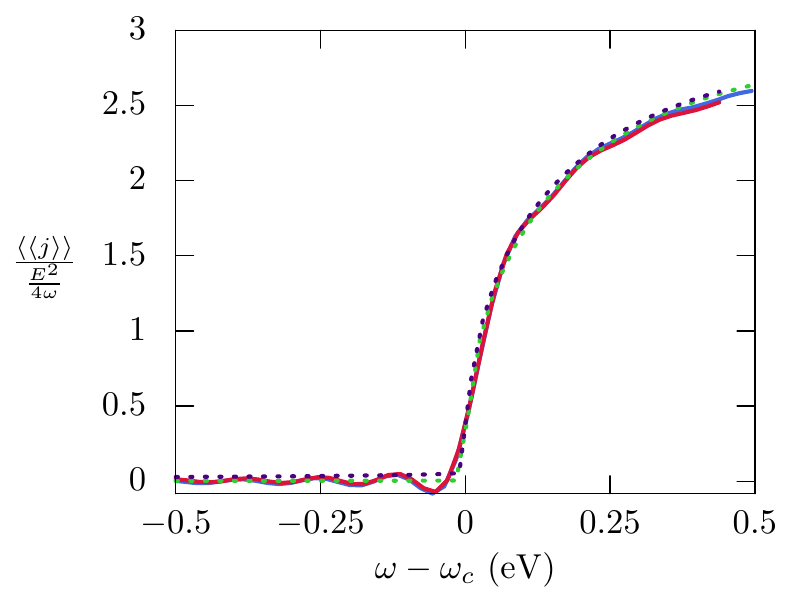}
  \caption{
    An average of the current after a number of periods as a function of $\omega-\omega_c$, for various values of the field: $E=3\ \mathrm{V}\cdot\mathrm{nm}^{-1}$ (blue), $E=10\ \mathrm{V}\cdot\mathrm{nm}^{-1}$ (red).
    For the sake of comparison, we have also plotted the asymptotic current predicted in \cite{FKS05} as dotted lines: green for $E=3\ \mathrm{V}\cdot\mathrm{nm}^{-1}$ and purple for $E=10\ \mathrm{V}\cdot\mathrm{nm}^{-1}$.
    All four curves are almost on top of each other.
    We see a sharp transition as $\omega$ crosses the critical frequency $\omega_c=U-\frac{k^2}2+\frac{E^2}{4\omega_c^2}$.
  }
  \label{Fig3}
\end{figure}
\bigskip

\noindent{\bf Acknowledgements:}
The authors wish to thank David Huse for valuable discussions, as well as the Institute for Advanced Study for its hospitality.
OC was partially supported by the NSF grants DMS-1515755 and DMS-2206241.
OC, RC, IJ and JLL were partially supported by AFOSR Grant FA9550-16-1-0037.
IJ was partially supported by NSF Grant DMS-1802170 and by a grant from the Simons Foundation, Grant Number 825876.

\vfill
\eject

\bibliographystyle{apsrev4-2_mod}
\bibliography{bibliography}

\end{document}